\documentclass[a4paper,onecolumn,accepted=2021-10-21]{quantumarticle}
\pdfoutput=1
\usepackage{amsmath} 					
\usepackage{amssymb}
\usepackage{amsthm}
\usepackage{array}       				
\usepackage{graphicx}  					
\usepackage{xcolor}   
\usepackage{calc}    			
\usepackage{makeidx}   					
\usepackage[utf8]{inputenc} 			
\usepackage{url}     					
\usepackage[english]{babel}
\usepackage{csquotes}
\usepackage{braket}
\usepackage{nicefrac}
\usepackage{mathtools}
\usepackage{booktabs}
\usepackage{comment}
\usepackage{enumitem}
\usepackage{hyperref}
\usepackage[normalem]{ulem}
\usepackage{tabularx}
\usepackage{dsfont}

\usepackage[backend=bibtex,style=phys,biblabel=brackets]{biblatex}
\AtEveryBibitem{\clearfield{note}} 

\renewcommand{\thesection}{\arabic{section}}
\renewcommand{\thesubsection}{\thesection.\arabic{subsection}}
\renewcommand{\thesubsubsection}{\thesubsection.\arabic{subsubsection}}

\makeatletter
\renewcommand{\p@subsection}{}
\renewcommand{\p@subsubsection}{}
\makeatother

\AtBeginDocument{
\heavyrulewidth=.08em
\lightrulewidth=.05em
\cmidrulewidth=.03em
\belowrulesep=.65ex
\belowbottomsep=0pt
\aboverulesep=.4ex
\abovetopsep=0pt
\cmidrulesep=\doublerulesep
\cmidrulekern=.5em
\defaultaddspace=.5em
}

\newcommand*{\Tr}{\operatorname{Tr}}

\newcommand{\R}[0]{\mathbb{R}}

\newcommand{\F}[0]{\mathcal{F}}
\newcommand{\X}[0]{\mathcal{X}}
\newcommand{\Y}[0]{\mathcal{Y}}

\newcommand{\A}[0]{\mathcal{A}}

\newcommand{\teps}[0]{\Tilde{\varepsilon}}

\newcommand{\Fclass}[0]{\mathcal{F}_{\Omega}^B}
\newcommand{\Fclasstil}[0]{\tilde{\mathcal{F}}_{\Omega}^B}

\newcommand{\Fclasstrignew}[0]{\mathcal{H}_\Omega^{\tilde{B}}}
\newcommand{\empRad}[0]{\hat{\mathcal{R}}}

\providecommand{\myvec}[1]{\ensuremath{\boldsymbol{#1}}}

\providecommand{\aa}{\ensuremath{\myvec{a}}}

\providecommand{\cc}{\ensuremath{\myvec{c}}}

\providecommand{\ee}{\ensuremath{\myvec{e}}}

\providecommand{\tt}{\ensuremath{\myvec{t}}}

\providecommand{\ww}{\ensuremath{\myvec{w}}}
\providecommand{\xx}{\ensuremath{\myvec{x}}}

\providecommand{\zz}{\ensuremath{\myvec{z}}}

\providecommand{\aalpha}{\ensuremath{\myvec{\alpha}}}

\providecommand{\ttheta}{\ensuremath{\myvec{\theta}}}

\providecommand{\xxi}{\ensuremath{\myvec{\xi}}}

\providecommand{\ssigma}{\ensuremath{\myvec{\sigma}}}

\providecommand{\oomega}{\ensuremath{\myvec{\omega}}}

\providecommand{\calB}{\ensuremath{\mathcal{B}}}
\providecommand{\calC}{\ensuremath{\mathcal{C}}}
\providecommand{\calD}{\ensuremath{\mathcal{D}}}

\providecommand{\calF}{\ensuremath{\mathcal{F}}}
\providecommand{\calG}{\ensuremath{\mathcal{G}}}
\providecommand{\calH}{\ensuremath{\mathcal{H}}}

\providecommand{\calM}{\ensuremath{\mathcal{M}}}
\providecommand{\calN}{\ensuremath{\mathcal{N}}}
\providecommand{\calO}{\ensuremath{\mathcal{O}}}

\providecommand{\calS}{\ensuremath{\mathcal{S}}}

\providecommand{\calX}{\ensuremath{\mathcal{X}}}
\providecommand{\calY}{\ensuremath{\mathcal{Y}}}
\providecommand{\calZ}{\ensuremath{\mathcal{Z}}}

\providecommand{\bbE}{\ensuremath{\mathbb{E}}}

\providecommand{\bbI}{\ensuremath{\mathbb{I}}}

\providecommand{\bbR}{\ensuremath{\mathbb{R}}}

\newtheorem{theorem}{Theorem}
\newtheorem{lemma}[theorem]{Lemma}

\newtheorem{corollary}[theorem]{Corollary}
\newtheorem{conjecture}[theorem]{Conjecture}

\theoremstyle{definition}
\newtheorem{definition}[theorem]{Definition}

\newtheorem{remark}[theorem]{Remark}

\newcommand{\norm}[1]{\left\lVert#1\right\rVert}

\renewcommand{\aa}{\ensuremath{\myvec{a}}}

\newcommand{\stkout}[1]{\ifmmode\text{\sout{\ensuremath{#1}}}\else\sout{#1}\fi}
\newif\ifverbose
\verbosetrue

\begin{document}

\author{Matthias~C.~Caro}
\affiliation{Department of Mathematics, Technical University of Munich, 85748 Garching, Germany}
\affiliation{Munich Center for Quantum Science and Technology (MCQST), 80799 Munich, Germany}

\author{Elies~Gil-Fuster}
\affiliation{Dahlem Center for Complex Quantum Systems, Freie Universit\"{a}t Berlin, 14195 Berlin, Germany}
\affiliation{Fraunhofer Heinrich Hertz Institute, 10587 Berlin, Germany}

\author{Johannes~Jakob~Meyer}
\affiliation{Dahlem Center for Complex Quantum Systems, Freie Universit\"{a}t Berlin, 14195 Berlin, Germany}

\author{Jens~Eisert}
\affiliation{Dahlem Center for Complex Quantum Systems, Freie Universit\"{a}t Berlin, 14195 Berlin, Germany}
\affiliation{Fraunhofer Heinrich Hertz Institute, 10587 Berlin, Germany}
\affiliation{Helmholtz-Zentrum Berlin f{\"u}r Materialien und Energie, 14109 Berlin, Germany}

\author{Ryan~Sweke}
\affiliation{Dahlem Center for Complex Quantum Systems, Freie Universit\"{a}t Berlin, 14195 Berlin, Germany}

\date{2021-10-28}
\title{Encoding-dependent generalization bounds
for \newline parametrized quantum circuits}

\begin{abstract}
A large body of recent work has begun to explore the potential of parametrized quantum circuits (PQCs) as machine learning models, within the framework of hybrid quantum-classical optimization. In particular, theoretical guarantees on the out-of-sample performance of such models, in terms of generalization bounds, have emerged. However, none of these generalization bounds depend explicitly on how the classical input data is encoded into the PQC. We derive generalization bounds for PQC-based models that depend explicitly on the strategy used for data-encoding. These imply bounds on the performance of trained PQC-based models on unseen data. Moreover, our results facilitate the selection of optimal data-encoding strategies via structural risk minimization, a mathematically rigorous framework for model selection. We obtain our generalization bounds by bounding the complexity of PQC-based models as measured by the Rademacher complexity and the metric entropy, two complexity measures from statistical learning theory. To achieve this, we rely on a representation of PQC-based models via trigonometric functions. Our generalization bounds emphasize the importance of well-considered data-encoding strategies for PQC-based models.
\end{abstract}

\maketitle

\section{Introduction}\label{s:intro}

Recent years have witnessed a surge of interest in the question of whether and how quantum computers can meaningfully address computational problems in machine learning~\cite{dunjko2018machine,biamonte2017quantum}.
This development has been largely driven by two factors. On the one hand, there is evidence that some quantum machine learning algorithms may lead to an increased performance over classical algorithms for the analysis of classical data with respect to important figures of merit~\cite{DeWolf,Wiebe,Seth,PACLearning,liu2020rigorous}. On the other hand, the increasing availability of quantum computational devices provides significant stimulus. While these ``noisy intermediate-scale quantum" (NISQ) devices are still a far cry from full-scale fault-tolerant quantum computers, there exists growing evidence that they may be able to out-perform classical computers on some highly-tailored tasks~\cite{GoogleSupremacy}.
Given the inherent limitations of NISQ devices, most current approaches to near-term quantum-enhanced machine learning fall under the umbrella of hybrid quantum-classical algorithms~\cite{bharti2021noisy}. Of particular prominence are variational quantum algorithms in which a \textit{parametrized quantum circuit} (PQC) is used to define a machine learning model which is then updated via a classical optimizer~\cite{McClean_2016,cerezo2020variational,Benedetti_2019}.

There is a wealth of architectural choices for PQC-based machine learning models. These include the width and depth of the quantum circuit, the precise layout and structure of trainable gates, as well as the mechanism via which classical data is encoded into the quantum circuit.
The flexibility in design choices for PQCs is often only perceived strongly in terms of the structure and layout of the trainable gates~\cite{Sim2019, hubregtsen2021evaluation}.
However, when using a PQC to define a machine learning model for \textit{classical data}, the data-encoding strategy becomes a necessary architectural design choice, which has received comparably little attention. Despite this, it has recently been shown that the data-encoding strategy is directly related to the expressive power of PQC-based models~\cite{gil_vidal2020input,schuld2021kernelmethods,PerezSalinas2020}.
In this work, we further the study of data-encoding strategies for PQC-based supervised learning models by investigating the effect of data-encoding strategies on \textit{generalization} performance. 

More specifically, we consider the following fundamental question: Given a PQC-based model which has been trained on a specific data set, can we place any guarantees on its expected \textit{out-of-sample} performance, i.e., its expected accuracy on new data, drawn from the same distribution as the training set? This question is motivated by the key insight that one should \textit{not} choose the model or architecture which performs best on the available training data, but rather the model for which one expects the best out-of-sample performance. Typically, one refers to the difference between the accuracy of a model on a given training set and its expected out-of-sample accuracy as the \textit{generalization gap}. We call a (probabilistic) upper bound on this generalization gap a \textit{generalization bound}. Historically, techniques for both proving generalization bounds and for using generalization bounds for principled model selection have been developed under the umbrella of \textit{statistical learning theory}~\cite{bishop2006pattern,scholkopf2002learning,MohriRostamizadehTalwalkar18}.

We start by presenting a selection of central notions in statistical learning theory. Of particular interest is the relation between generalization bounds and \textit{complexity measures} of different types.
Indeed, due to a large body of existing literature,
bounding the generalization gap of a learning model typically reduces to bounding some quantifiable property of the hypothesis class used for learning.
There are many examples of such complexity measures (also known as \emph{capacity metrics} or just \emph{expressivity measures}), and based on their specifics they are used for different learning models, either quantum or not.
In this work, we employ generalization bounds based on the Rademacher complexity and the metric entropy. However, we want to mention that there are also other important approaches to generalization not taken here, such as stability~\cite{bousquet2002stability}, compression~\cite{littlestone1986relating}, or the PAC-Bayesian framework~\cite{mcallester1999some}.

Given the fundamental role of generalization bounds, there has recently been a strong and steady stream of works contributing to the derivation of generalization bounds for PQC-based models~\cite{CaroDatta2020, abbas2020power, Bu.2021a, Bu.2021b, Bu.2021c, du2021efficient, huang2021power, banchi2021generalization, Vedran2021}. However, as discussed in detail in Section~\ref{s:prior_work}, these prior works all differ from our results in a variety of ways. Firstly, they considered only ``encoding-first" PQC architectures, in which the PQC-based models are assumed to consist of an initial data-encoding block, mapping a classical input to a data-dependent quantum state, followed by a circuit consisting only of fixed and trainable gates. In contrast, we consider PQC-based models incorporating \textit{data re-uploading}~\cite{PerezSalinas2020}, in which trainable circuit blocks are interleaved with data-encoding circuit blocks. This is particularly relevant given the results of Refs.~\cite{SchuldSwekeMeyer2021,gil_vidal2020input}, which have illuminated the significant effects of data re-uploading on the expressive power of PQC-based models.

Additionally, our work is the first to provide a generalization bound from which it is immediately clear how altering the data-encoding strategy influences the generalization performance of the model. This is possible because our bound depends \emph{explicitly} on architectural hyper-parameters associated with the data-encoding strategy. This sets our results apart from prior art where the data-encoding figured only \emph{implicitly}, if at all. We discuss this difference between implicitly and explicitly encoding-dependent generalization bounds more concretely in Section~\ref{s:prior_work}.

In order to obtain our generalization bounds, we rely strongly on a representation of PQC-based models via generalized trigonometric polynomials (GTPs), which has been 
previously derived in Refs.~\cite{SchuldSwekeMeyer2021,gil_vidal2020input}. In particular, we exploit the fact that the data-encoding strategy of the PQC-based model directly determines the frequency spectrum of the corresponding GTPs. As such, the number of accessible frequencies in the GTP representation provides a natural measure of the complexity of a particular data-encoding strategy. Given this, we first derive generalization bounds for GTPs, which exhibit a dependence on the square root of the number of accessible frequencies. We then proceed to determine, for different data-encoding strategies, upper bounds on the number of accessible frequencies in the GTP representation. We use these results to identify a variety of natural data-encoding strategies for which the number of accessible frequencies, and therefore the associated generalization bounds, scale polynomially with the number of data-encoding gates. While one \textit{cannot} use generalization bounds alone to recommend an optimal data-encoding strategy, we discuss how these generalization bounds can be combined with empirical risk estimates, via \textit{structural risk minimization}, to facilitate the selection of an optimal data-encoding strategy for a given problem.

\subsection{Structure of this work}

\begin{figure}
    \centering
    \includegraphics[width=\linewidth]{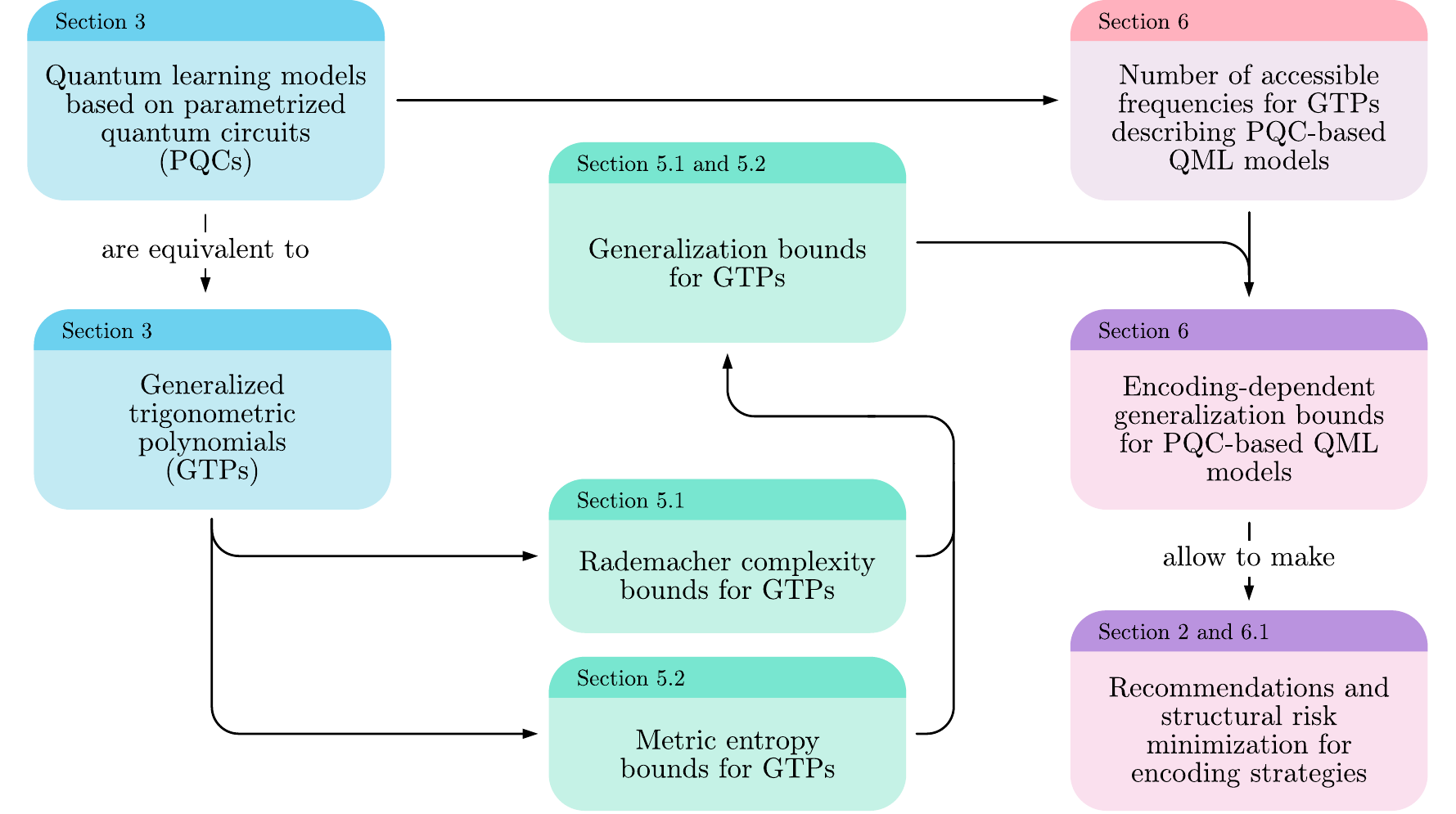}
    \caption{A flowchart of the argument presented in this work.}
    \label{f:flowchart}
\end{figure}

This work is structured as follows: Section~\ref{s:motivation} gives a pedagogical introduction to statistical learning theory, explains the importance of generalization bounds, and discusses the structural risk minimization principle. After establishing these concepts, we formulate the main questions addressed in this work. In Section~\ref{s:pqc_models}, we begin by introducing the PQC-based learning models used in this work. We then present a detailed discussion of the approach of Ref.~\cite{SchuldSwekeMeyer2021}, which demonstrates how the functions implemented by a PQC-based model can be represented by generalized trigonometric polynomials. In particular, we emphasize how the data encoding strategy of the PQC-based model translates to the accessible
frequencies of the generalized trigonometric polynomials. Section~\ref{s:prior_work} then provides a detailed review of prior work on generalization in quantum machine learning. In Section~\ref{s:results}, we establish generalization bounds for classes of generalized trigonometric polynomials in terms of the number of accessible frequencies. We present one approach via the Rademacher complexity (Section~\ref{sbs:gen-bounds-trig-poly-rademacher}) and another via covering numbers (Section~\ref{sbs:gen-bounds-trig-poly-covering}).
Section~\ref{s:gen-bounds-pqc} then expands upon Section~\ref{s:pqc_models} by deriving upper bounds on the number of accessible frequencies, in the generalized trigonometric polynomial representation of the PQC-based models associated with different data-encoding strategies. This analysis allows us to use the results from Section~\ref{s:results} to state explicitly encoding-dependent generalization bounds for PQC-based models, and to compare different encoding strategies from a generalization perspective. We discuss the implications of our results in Section~\ref{s:discussion}. In particular, we emphasize how our results are complementary to many prior works, but also describe how the different approaches can be combined. Additionally, we sketch some directions for future research. Section~\ref{s:conclusion} contains a short summary of our work. The logical flow of this manuscript is visualized in Figure~\ref{f:flowchart}.

\section{Motivation: Generalization bounds, sample complexities and model selection}\label{s:motivation}

To motivate the content of this work and to define the setting, we start with a brief and select introduction to the framework of \emph{statistical learning theory}. Interested readers are referred to Refs.~\cite{MohriRostamizadehTalwalkar18} and~\cite{shalev-shwartz_ben-david_2014} for a more detailed presentation. Within this framework, any supervised learning problem is defined by a \emph{domain} $\mathcal{X}$, a \emph{co-domain} $\mathcal{Y}$, a \emph{probability distribution} $P$ over $\mathcal{X}\times\mathcal{Y}$ and a \emph{loss function} $\ell:\mathcal{Y}\times\Y \rightarrow \mathbb{R}$. We assume that $\mathcal{X},\mathcal{Y}$ and $\ell$ are known, while $P$ is unknown. We will denote the set of all functions from $\X$ to $\Y$ as $\Y^\X$. To gain intuition, it is useful to think of the situation in which there exists a deterministic rule for assigning predictions to domain elements. We can model this in the framework outlined above with an unknown target function $f\in\Y^\X$, as well as some unknown probability distribution $P_\mathcal{X}$ over $\mathcal{X}$, such that samples from $P$ are obtained by first drawing a domain element $\xx\in\mathcal{X}$ from $P_\X$, and then outputting the tuple $(\xx,f(\xx))$, i.e.
\begin{align}
    P(\xx,y) = \begin{cases}
    P_{\X}(\xx)\quad&\text{ if } y = f(\xx),\\
    0&\text{ if } y\neq f(\xx).
    \end{cases}
\end{align}
In general, however, it may be the case that there exists $y_1\neq y_2$ for which both $P(\xx,y_1)>0$ and $P(\xx,y_2)>0$, i.e., that the underlying process for labeling data points is not deterministic. 

Additionally, we are given a training data set
\begin{equation}
    S = \{(\xx_i,y_i)\sim P\,|\,i\in\{1,\ldots,m\}\}
\end{equation}
of $m$ tuples drawn independently from (the unknown distribution) $P$, and our goal is to design a \emph{learning algorithm} $\mathcal{A}$ which, given $S$ as input, outputs a hypothesis $h \in \Y^{\X}$ that achieves a sufficiently small \textit{risk}
\begin{equation}\label{e:expected_risk}
    R(h) = \int_{\X\times\Y}\ell(y,h(\xx))~\mathrm{d}P(\xx,y).
\end{equation}
Informally, we often refer to the risk $R(h)$ as characterizing the \textit{out-of-sample} performance of the hypothesis $h$, as it is this quantity which tells us how well we can expect the hypothesis $h$ to perform on (possibly previously unseen) future data drawn from $P$. It is critical to note, however, that as the underlying probability distribution $P$ is unknown, given a hypothesis $h\in \Y^\X$, one \textit{cannot} directly evaluate $R(h)$. In light of this, a natural alternative is to evaluate the \textit{empirical risk} of $h$ with respect to $S$, which is defined as the average loss over the training samples
\begin{equation}\label{e:empirical_risk}
    \hat{R}_S(h) = \frac{1}{|S|}\sum_{(\xx_i,y_i)\in S} \ell(y_i,h(\xx_i)).
\end{equation}
In contrast to the risk $R(h)$, the empirical risk $\hat{R}(h)$ characterizes the \textit{in-sample} performance of $h$ with respect to the data set $S$, which has been sampled from $P$. 

Naively, one might hope to be able to construct learning algorithms which could in principle output \textit{any} $h\in\Y^\X$. However, the ``no-free-lunch" theorem rules out the possibility of meaningful learning in this case~\cite{Wolf.2020}, and therefore we typically consider learning algorithms whose range is some subset $\mathcal{F} \subseteq \Y^\X$. We then refer to $\mathcal{F}$ as the \textit{hypothesis class} associated with the learning algorithm which is, by assumption, also known to the learning algorithm. 
To gain some intuition, one could think of $\mathcal{F}$ as the set of all functions realizable by neural networks of some fixed width and depth, or, as we describe in Section~\ref{s:pqc_models}, as the set of all functions realizable by a parametrized quantum circuit model with some fixed architecture. With respect to this setting, the following natural question arises: Suppose we have a learning algorithm $\mathcal{A}$ with hypothesis class~$\mathcal{F}$, which has been run on a randomly drawn data set of $m$ samples $S\sim P^m$ and outputs some hypothesis $h\in\mathcal{F}$, as well as some ``training log" which we denote by $\mathrm{hist}(\mathcal{A},S)$\footnote{Such a training log could for example record the value of the empirical risk, or properties of the trial hypotheses (such as weight matrices for neural networks), at each stage of an iterative optimization procedure.}. Given the achieved empirical risk \smash{$\hat{R}_S(h)$}, can we put an upper bound on the true risk $R(h)$, which holds with high probability over the randomly drawn data set $S$? More specifically, can we make a statement of the form: For all $\delta \in (0,1)$, with probability $1-\delta$ over $S\sim P^m$, for all $h\in\F$  we have that 
\begin{equation}\label{e:gen_bound}
    R(h) \leq \hat{R}_S(h) + g(\mathcal{F},h,m,S,\mathcal{A},\mathrm{hist}(\A,S),\delta).
\end{equation}
We refer to such a statement as a \textit{generalization bound}, and note that the function $g$ appearing in Eq.~\eqref{e:gen_bound} provides a (probabilistic) upper bound on the quantity \smash{$R(h) - \hat{R}_S(h)$}, which we call \textit{generalization gap} (of $h$ with respect to~$S$). 
Such bounds are desirable because they allow us to leverage the information we have access to -- i.e., the empirical risk, and properties of the learning algorithm, data set and optimization procedure -- to upper bound $R(h)$, which is the quantity we do not have access to, but are ultimately interested in.
In general, as indicated explicitly in Eq.~\eqref{e:gen_bound}, the upper bound $g$ on the generalization gap could depend on properties of the achieved hypothesis $h$, properties of the data set~$S$, properties of the learning algorithm $\mathcal{A}$, and details of the optimization that led to $h$. However, in this work we will focus on \textit{uniform} generalization bounds of the form: for all $\delta \in (0,1)$, with probability $1-\delta$ over $S\sim P^m$, we have for all $h\in\F$ that
\begin{equation}\label{eqn:gen_bound_g}
    R(h) \leq \hat{R}_S(h) + g(\mathcal{F},m,\delta).
\end{equation}
To be specific, we focus on generalization bounds for which the upper bound on the generalization gap -- i.e., the function $g$ -- depends only on properties of the hypothesis class $\mathcal{F}$, the data set size $m$ and the desired probability $\delta$. We note that the term ``uniform" is used when describing such generalization bounds to indicate that, with respect to a fixed data set size $m$ and probability threshold $\delta$, the upper bound on the generalization gap will be the same -- i.e., uniform -- for all $h\in\mathcal{F}$. While it is known that there exist scenarios in which uniform generalization bounds are not tight~\cite{zhang,jiang}, we postpone a discussion of these issues to Section~\ref{s:discussion}.

As motivated above, given a uniform generalization bound for a hypothesis class $\mathcal{F}$, one typical application is as follows: Given a data set $S$ sampled from $P$, with $|S|=m$, run some learning algorithm to obtain a hypothesis $h\in\mathcal{F}$, evaluate its empirical risk \smash{$\hat{R}_S(h)$}, and then use the generalization bound to place a (probabilistic) upper bound on the true risk $R(h)$. 
However, we can also often straightforwardly use such a generalization bound to answer the following natural question: Given some $\epsilon>0$ and some $\delta\in (0,1)$, what is the minimum size of $S$ sufficient to ensure that, with probability $1-\delta$, for all $h\in\F$, the generalization gap satisfies \smash{$R(h) - \hat{R}_S(h)\leq \epsilon$}? To see this, note that if we have a uniform generalization bound, then by setting
\begin{equation}\label{e:g_upper}
    g(\mathcal{F},m,\delta) \leq \epsilon
\end{equation}
and solving for $m$, it is often possible to find some function $f(\epsilon,\delta,\mathcal{F})$ such that, with probability $1-\delta$ over $S\sim P^m$,
\begin{equation}\label{e:sc}
    m \geq f(\epsilon,\delta,\mathcal{F}) \Rightarrow \forall h\in\mathcal{F}: R(h) - \hat{R}_S(h) \leq  g(\mathcal{F},m,\delta) \leq \epsilon.
\end{equation}
As the generalization bound may not be tight, we therefore see that $f(\epsilon,\delta,\mathcal{F})$ provides an \textit{upper bound} on the minimum size of $S$ sufficient to probabilistically guarantee a generalization gap less than $\epsilon$ for all $h\in\mathcal{F}$.

Finally, apart from the fundamental applications of allowing us to bound the out-of-sample performance of a hypothesis, or upper bound the minimum sample-size sufficient to guarantee a certain generalization gap, generalization bounds also allow us to address the issue of \textit{model selection}, via the framework of \textit{structural risk minimization}~\cite{MohriRostamizadehTalwalkar18}. Importantly, we note that one \textit{cannot} simply use only the function $g(k,m,\delta)$ for model selection: A trivial learning model, which outputs the same hypothesis independently of the input data, has $g(k,m,\delta)=0$, but
cannot achieve good prediction performance on interesting tasks.
Structural risk minimization thus suggests combining a generalization bound with an empirical risk evaluation on a specific data-set to choose the model with the smallest upper-bound on the true risk. More specifically, let us assume that our hypothesis class depends on some ``architectural hyper-parameter" $k$, with some notion of ordering such that
\begin{equation}\label{e:increasing_complexity}
    k_1 \leq k_2 \implies \mathcal{F}_{k_1} \subseteq \mathcal{F}_{k_2}.
\end{equation}
For example, $\mathcal{F}_k$ could be the set of all neural networks of fixed width and depth $k$. Given this, how should we choose the hypothesis class -- or model complexity -- that we use for a given learning problem? As illustrated in Figure~\ref{f:SRM}, generalization bounds, when combined with empirical risk evaluations, can allow us to answer this question. In particular, assume that we have a uniform generalization bound of the form: For all $\delta \in (0,1)$, with probability $1-\delta$ over $S\sim P^m$, for all $h\in\F_k$,
\begin{equation}\label{e:gen_bound_for_SRM}
    R(h) \leq \hat{R}_S(h) + g(k,m,\delta),
\end{equation}
where $g(k,m,\delta)$ is non-decreasing with respect to $k$. Here, we have written $g(k,m,\delta)$ rather than $g(\F_k,m,\delta)$ to emphasize the assumption that the hyper-parameter $k$ is the only property of $\F_k$ on which the generalization bound depends explicitly. 
While increasing $k$ increases the expressivity of the hypothesis class and therefore typically leads to smaller empirical risk, it also increases the upper bound $g(k,m,\delta)$ on the generalization gap and may therefore lead to hypotheses with worse out-of-sample performance. As such, a natural strategy to find an optimal hypothesis -- in the sense of having the smallest probabilistic upper bound on the true risk 
-- is as follows:
\begin{figure}
    \centering
    \includegraphics[width=0.52\textwidth]{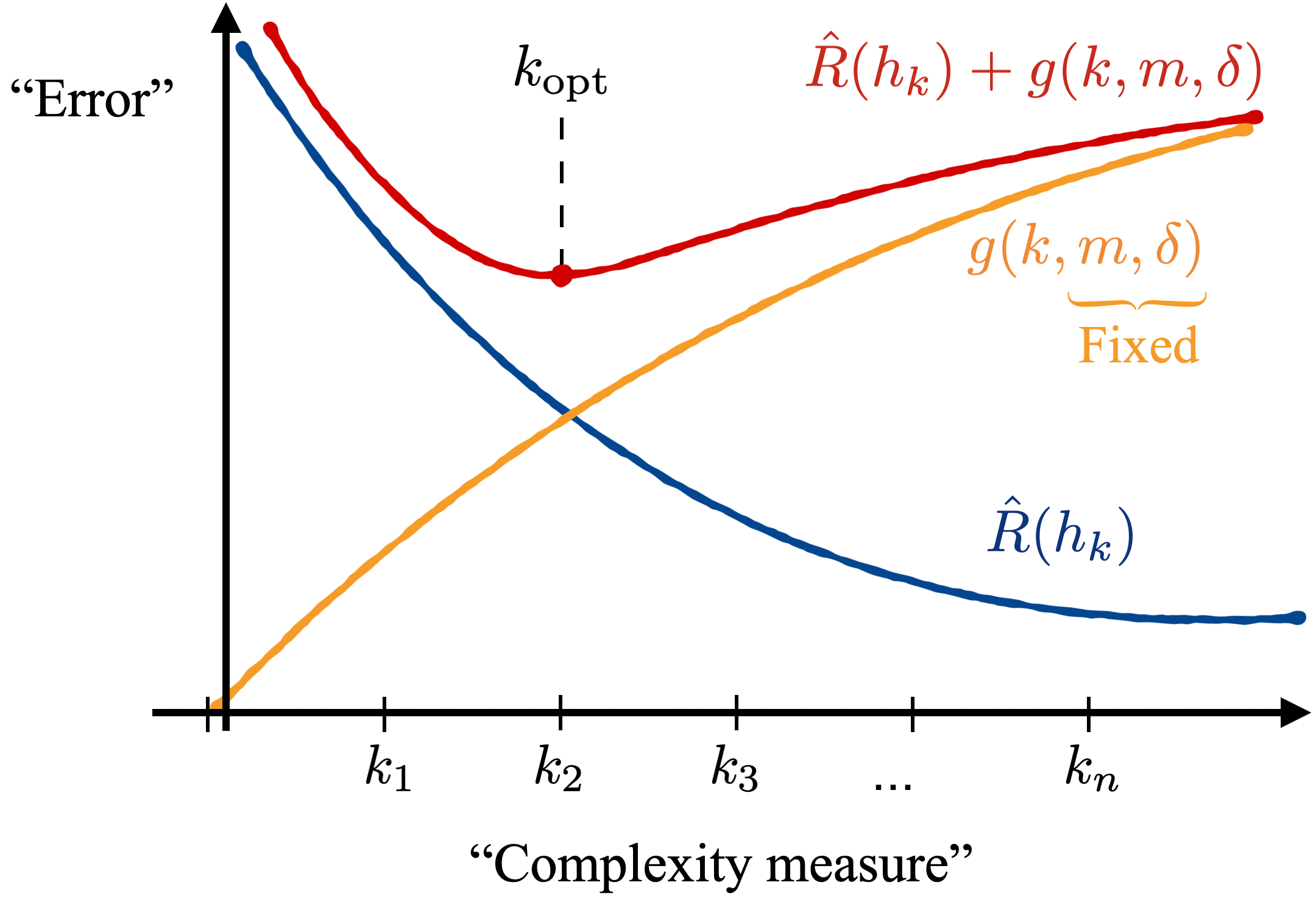}
    \caption{Illustration of structural risk minimization (adapted from Ref.~\cite{MohriRostamizadehTalwalkar18}). Increasing the complexity of a hypothesis class typically allows one to obtain hypotheses with decreasing empirical risk. However, in many cases increasing the complexity of a hypothesis class also leads to a larger upper bound on the generalization gap. Structural risk minimization aims to identify a hypothesis with the smallest upper bound on the true risk that quantifies the out-of-sample performance by combining an evaluation of the empirical risk of candidate hypotheses with an upper bound on the generalization gap of the relevant hypothesis class.}
    \label{f:SRM}
\end{figure}
\begin{enumerate}
    \item For $k$ in $\{k_1,\ldots,k_n\}$, run the learning algorithm $\mathcal{A}_k$, with hypothesis class $\mathcal{F}_k$, and obtain the hypothesis $h_k$.
    \item Calculate $k_\mathrm{opt} = \mathrm{argmin}_{k}[\hat{R}_S(h_k) + g(k,m,\delta)]$.
    \item Output $h_{k_{\mathrm{opt}}}$.
\end{enumerate}
We refer to such a procedure as \textit{structural risk minimization}\footnote{We note that the term ``structural risk minimization" is sometimes used to refer to the strategy of minimizing a regularized empirical risk, with an additive regularization term which penalizes high model complexity. However, we follow Ref.~\cite{MohriRostamizadehTalwalkar18} in our definition and presentation.}, and contrast this with \textit{empirical risk minimization}, which simply outputs the hypothesis minimizing the empirical risk.
In light of the above discussion, we note that, given a family of hypothesis classes $\{\calF_{k}\}$,
each specified by some architectural hyper-parameter $k$ and satisfying the condition of Eq.~\eqref{e:increasing_complexity},
we would ideally like to obtain an upper bound on the generalization gap $g(k,m,\delta)$ which grows \textit{slowly} with respect to $k$. In particular, we can now understand this from two 
different but complementary perspectives:

Firstly, from the structural risk minimization (or model selection) 
perspective, we see from Figure~\ref{f:SRM} that slow growth of $g(k,m,\delta)$ is indicative of our ability to exploit the expressivity of more complex hypothesis classes, i.e.\ those with larger $k$, without risking poor generalization performance due to overfitting. More specifically, under the assumption of monotonically decreasing empirical risk, the slower $g(k,m,\delta)$ grows, the longer we can expect the quantity \smash{$\hat{R}_S(h_k) + g(k,m,\delta)$} to decrease before reaching a minimum, and therefore the smaller we can expect our ultimate upper bound on the true risk of the optimal hypothesis $h_{k_{\mathrm{opt}}}$ to be. In contrast, if $g(k,m,\delta)$ grows too fast with respect to $k$, then even if we can achieve very small empirical risk by increasing model complexity, we do not expect to be able to achieve a sufficiently small upper bound on the true risk of the optimal hypothesis $h_{k_{\mathrm{opt}}}$. 

Secondly, from the sample complexity perspective, let us denote by $f(\epsilon,\delta,k)$ the complementary upper bound on the minimum sample sample size $m$ sufficient to probabilistically ensure a generalization gap less than $\epsilon>0$, which typically follows from $g(k,m,\delta)$ 
(as we recall from the discussion around Eqs.~\eqref{e:g_upper} and \eqref{e:sc}). 
As we naturally expect $g(k,m,\delta)$ to be decreasing with increasing $m$, slow growth of $g(k,m,\delta)$ with respect to $k$ typically implies slow growth of $f(\epsilon,\delta,k)$ with respect to $k$. In other words, slow growth of $g(k,m,\delta)$ typically implies slow growth, with respect to model complexity, of the minimum amount of data one has to use before being able to probabilistically guarantee a certain generalization gap for all output hypotheses. As generating data (i.e., sampling from
the distribution $P$) may be expensive or difficult, and as the run-time of learning algorithms typically scales with respect to the data set size, slow growth of $g(k,m,\delta)$ therefore facilitates the process of learning with models of higher complexity.

Given the above observations, we can finally understand the motivation of this work in an informal way. In particular, in the following section we will see that parametrized quantum circuits (PQCs) naturally give rise to hypothesis classes with multiple architectural hyper-parameters, each reflecting a different aspect of the circuit architecture, such as circuit depth, circuit width, the total number of gates or the total number of data-encoding gates of a particular type. In Section~\ref{s:prior_work} we will then see that a body of previous work has resulted in a collection of generalization bounds for PQC-based models, each of which depend explicitly on some subset of architectural hyper-parameters, but not on others. As of yet, however, there exist no generalization bounds which depend explicitly on hyper-parameters associated with the data-encoding strategy, despite the important role such strategies play in determining the expressive power of PQC-based hypothesis classes~\cite{SchuldSwekeMeyer2021}.
As such, the questions which we address in this work are as follows:

\begin{enumerate}[label=(\alph*),font=\itshape]
\item \textit{Can we derive generalization bounds for PQC-based hypothesis classes which depend explicitly on hyper-parameters associated with the data-encoding strategy?}
\item \textit{Can we use such bounds to identify data-encoding strategies for which the upper bounds on the generalization gap grow polynomially with respect to the architectural hyper-parameter relevant to the encoding strategy?}
\end{enumerate}

As will be discussed in Section~\ref{s:discussion}, apart from filling a gap in our understanding of the manner in which the data-encoding influences generalization, such bounds would also complement existing works, in that they would allow one to perform structural risk minimization with respect to multiple architectural hyper-parameters simultaneously. With this motivation in mind, before proceeding it is worth briefly mentioning \textit{how} (uniform) generalization bounds are typically obtained. Intuitively, one might expect that the generalization performance of a hypothesis class is related to how \textit{complex} (or how \textit{expressive}) the hypothesis class is, and thus one might hope for the existence of a complexity measure for hypothesis classes from which generalization bounds follow. This intuition is indeed correct, and in fact a large amount of work in statistical learning theory has resulted in a variety of suitable complexity measures -- such as the VC dimension~\cite{Vapnik.1971}, Rademacher complexity~\cite{Bartlett.2002}, pseudo-dimension~\cite{Pollard.1984} and metric-entropy amongst others -- all of which directly give rise to generalization bounds~\cite{MohriRostamizadehTalwalkar18, shalev-shwartz_ben-david_2014,Wolf.2020}. As a result, given a hypothesis class $\mathcal{F}_k$, one typically proves a uniform generalization bound for $\mathcal{F}_k$, which depends explicitly on the architectural hyper-parameter $k$, by 
first characterizing the dependence of a suitable complexity measure $C$ on $k$ (i.e., by writing/bounding ${C}(\mathcal{F}_k)$ explicitly in terms of $k$), and then
writing down the known generalization bound which follows from $C(\mathcal{F}_k)$. 
We also follow such a strategy in this work by first characterizing both the Rademacher complexity and metric-entropy of PQC-based models in terms of architectural hyper-parameters related to the data-encoding strategy and then presenting generalization bounds in terms of these complexity measures.
At this stage it is hopefully clear, both \textit{why} generalization bounds are desirable, and \textit{how} (at least intuitively) one might obtain such bounds. Given this, we proceed in the following section to define more precisely the PQC-based hypothesis classes considered in this work. 

\section{Parametrized quantum circuit based model classes}\label{s:pqc_models}
\begin{figure}
    \centering
\includegraphics[width=.8\columnwidth]{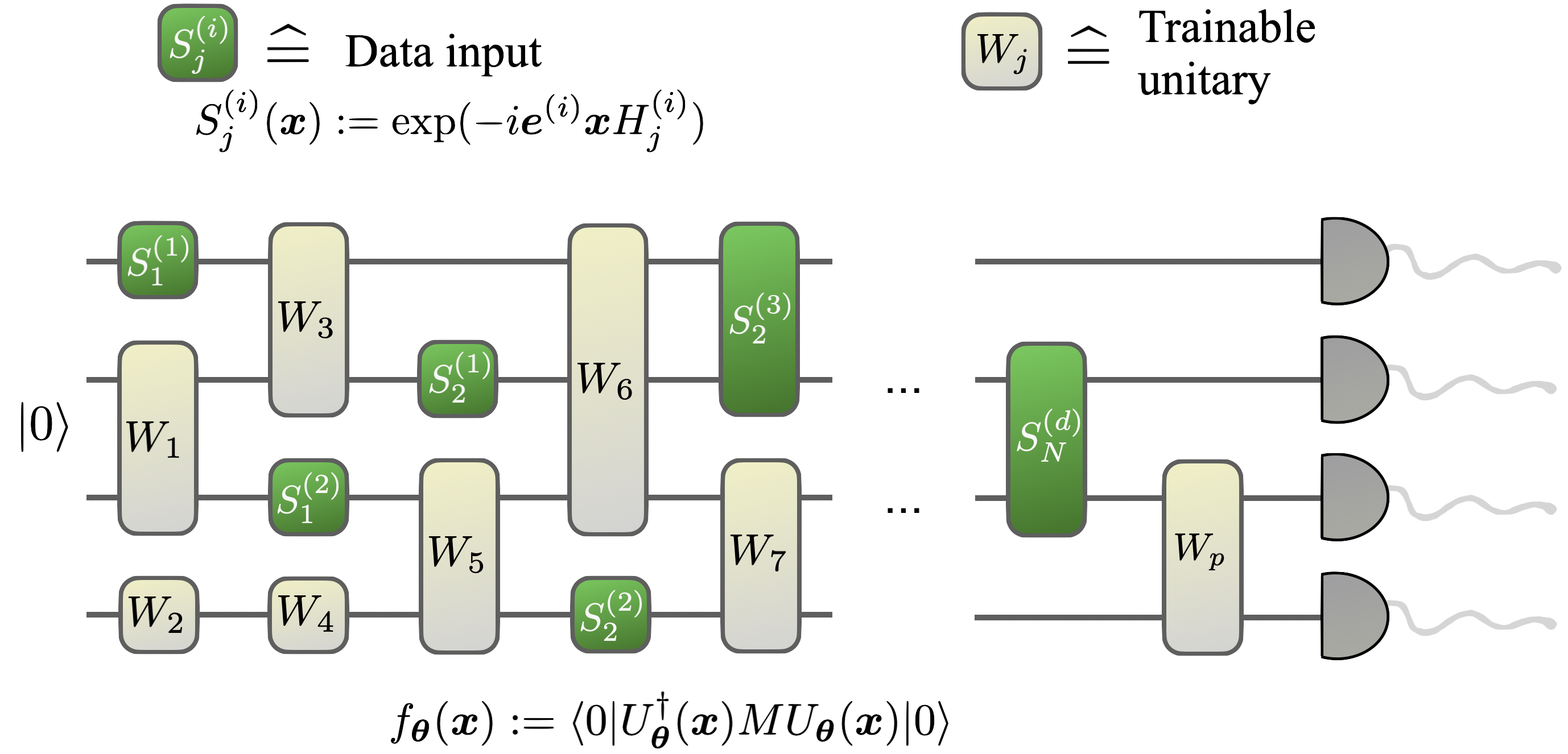}
    \caption{Circuit model considered in this work.
    We assume that the circuit consists of gates which are parametrized either by the data $\xx$ (data-encoding gates), or the trainable parameters $\ttheta$ (trainable gates). The data encoding gates are assumed to implement the time evolution of a data-encoding Hamiltonian, with evolution time given by some data coordinate $x^{(i)}=\ee^{(i)}\xx$.
    The model output is then given by the expectation value of an observable $M$.}
    \label{fig:freeform_circuit_model}
\end{figure}

\emph{Parametrized quantum circuits (PQCs)} are ubiquitous in the field of near-term quantum computing~\cite{bharti2021noisy, cerezo2020variational,McClean_2016} and can be used to construct quantum machine learning models~\cite{Benedetti_2019}. We will consider qubit-based quantum systems. The focus of this work lies on \emph{variational quantum machine learning models} that are constructed from a PQC $U_{\ttheta}(\xx)$ that depends on trainable parameters $\ttheta \in \Theta$ and on data inputs $\xx \in \calX$. A prediction in the co-domain $\calY = \bbR$ is then obtained by evaluating the expectation value of a fixed observable $M$, which can be efficiently evaluated, as
\begin{align}
    f_{\ttheta}(\xx) = \langle 0 | U^{\dagger}_{\ttheta}(\xx) M U_{\ttheta}(\xx) | 0\rangle.
\end{align}
In the following, we assume that the data inputs are $d$-dimensional real-valued vectors with entries in the interval $[0,2\pi)$, i.e., $\calX = [0,2\pi)^d$. This choice is somewhat arbitrary, as data can always be rescaled to fit into a particular interval.
However, $[0,2\pi)$ is a natural choice because quantum gates available on actual hardware are usually parametrized in terms of \emph{angles}. As will become apparent later, we need not make any assumptions on the nature of the trainable parameters, but in most cases they will also be angles, i.e., $\Theta = [0,2\pi)^p$, where $p$ is the number of trainable parameters.

We also make some assumptions on the structure of the circuit \smash{$U_{\ttheta}(\xx)$}. Our model is motivated by the actual quantum circuits that can be executed on NISQ devices. These devices usually only allow fixed gates and parametrized evolutions under device-specific Hamiltonians~\cite{Qiskit,cirq_developers2021cirq,bergholm2020pennylane}.
In our model, the data inputs $\xx$ and the trainable parameters $\ttheta$ enter the circuit through different gates. The unitaries parametrized by $\ttheta$, denoted by \smash{$ \{ W_{i}(\ttheta) \}$}, constitute the trainable part of the model. Fixed unitaries can be absorbed into the trainable unitaries. 

We assume that the gates through which the data enters the circuit are time evolutions under some Hamiltonian, where the \enquote{evolution time} is given by one of the data coordinates \smash{$x^{(i)}$}.
We denote the $j$-th gate that encodes the data coordinate $x^{(i)}$ as
\begin{align}\label{e:d_enc_gate}
    S^{(i)}_{j}(\xx) = \exp\left( -i x^{(i)} H_{j}^{(i)} \right)= \exp\left( -i \ee^{(i)} \xx H_{j}^{(i)} \right),
\end{align} 
where we rewrote the encoding gate in terms of the input data vectors by recognizing that \smash{$x^{(i)} = \ee^{(i)}\xx$}, where $\ee^{(i)}$ is a standard basis vector.
It is of course possible to consider more general dependencies of the evolution time on the input data, i.e.\ in terms of linear combinations or even non-linear functions of the data coordinates. However, we choose not to include models with such classical pre-processing of the data, in order to isolate the part of the model which is truly quantum.
Indeed, if one allowed for arbitrary pre-processing, then one could just use a very complicated neural network to find suitable evolution times for good predictions, but that would miss the point of using a quantum learning model at all. We note though that our definition still encompasses such approaches after a suitable reparametrization of the inputs, which will usually result in a larger number of input coordinates.

For our analysis, no restriction on the placement of the trainable gates and the data-encoding gates in the circuit is necessary. Thus, we assume that they can be arranged arbitrarily, as depicted in Figure~\ref{fig:freeform_circuit_model}.
However, we will refer to the choice of data-encoding Hamiltonians per data coordinate as \smash{$\calD^{(i)} = \{ H_j^{(i)} \}$} and call the union of these sets over all data coordinates the \emph{data-encoding strategy}
\begin{align}
    \calD = \big( \calD^{(1)}, \calD^{(2)}, \dots, \calD^{(d)} \big).
\end{align}
The total number of encoding gates per data coordinate is $N^{(i)} = |\calD^{(i)}|$ and the total number of data-encoding gates is $N = \sum_{i=1}^d N^{(i)}$.

A data-encoding strategy $\calD$ together with a fixed circuit structure and a choice of trainable gates defines a parametrized quantum circuit $U_{\ttheta}(\xx)$. We denote the fact that this circuit uses the encoding strategy $\calD$ as {$U_{\ttheta}(\xx) \sim \calD$}. When we fix an observable $M$ to generate the predictions, this defines a function class
\begin{align}
    \calF_{\Theta, \calD, M} \coloneqq \{ [0, 2\pi)^d \ni  \xx \mapsto \langle 0 | U^{\dagger}_{\ttheta}(\xx) M U_{\ttheta}(\xx) | 0\rangle \, | \, \ttheta \in \Theta, U_{\ttheta}(\xx) \sim \calD \},
\end{align}
which is obtained by considering all possible parametrizations $\ttheta \in \Theta$ of the trainable gates.
This function class depends explicitly on the parametrization of the trainable parts of the circuit and on the data-encoding strategy. 
As we ultimately want to obtain generalization bounds that depend on the hyper-parameters associated with the encoding strategy -- such as the number of encoding gates $N$ -- it will be helpful for us to reformulate the function class in a way that makes it more amenable to the analyses in the following sections.
To this end, we draw on the results of Refs.~\cite{gil_vidal2020input,SchuldSwekeMeyer2021}, which show that the nature of the data encoding gates as Hamiltonian evolutions allows us to expand the model output as a \emph{generalized trigonometric polynomial (GTP)}. A GTP \enquote{generalizes} the notion of a trigonometric polynomial by allowing arbitrary frequencies as in
\begin{align}
    f_{\ttheta}(\xx) = \sum_{\oomega \in \Omega(\calD)} c_{\oomega}(\ttheta, M) e^{-i \oomega \xx}.
\end{align}
While the GTP's coefficients $\{ c_{\oomega} \}$ depend on the particular parametrization and observable, the set of frequencies $\Omega(\calD)$ depends solely on the chosen data-encoding strategy $\calD$, in particular on the spectra of the Hamiltonians \smash{$\{ H_j^{(i)} \}$} that yield the data encoding evolutions \smash{$\{ S^{(i)}_j(\xx) \}$}. We describe the procedure for obtaining such a GTP representation in more detail below.
The fact that the expectation value is always real is reflected by $c_{\oomega} = c_{-\oomega}^{*}$ and by the observation that $\oomega \in \Omega(\calD)$ implies that also $-\oomega \in \Omega(\calD)$. 
Additionally, we note that the absolute value of any expectation value obtained from measuring $M$ is upper bounded by its operator norm $\lVert M \rVert_{\infty}$, and therefore, if we assume that $\lVert M \rVert_{\infty} \leq B$, then 
\begin{align}\label{eq:function_class_bounded_inf_norm}
    \calF_{\Theta, \calD, M} 
    \subseteq \Fclass 
    \coloneqq \left\{[0,2\pi)^d\ni \xx\mapsto f(\xx) = \sum\limits_{\oomega\in\Omega} c_{\oomega}\exp(-i \oomega \xx) \, \Big\vert \, (c_{\oomega})_{\oomega\in \Omega} \text{ such that } \norm{f}_{\infty} \leq B\right\},
\end{align}
where $\Omega = \Omega(\calD)$. 
We have thus defined a function class that solely depends on the data-encoding strategy.
We stress that this function class subsumes all possible ways to parametrize the trainable parts of a circuit with fixed data-encoding strategy $\calD$ and fixed observable $M$, but also goes beyond this by allowing all possible choices of observable $M$ such that $\lVert M \rVert_{\infty} \leq B$. Therefore, it also contains models where not only the parameters of the trainable gates, but also the measurement itself is subject to optimization. 
In going from $\calF_{\Theta, \calD, M}$ to $\Fclass$, we effectively allow for a universal trainable part and observable, which enables us to focus on the encoding strategy. Studying intermediate classes between $\calF_{\Theta, \calD, M}$ and $\Fclass$ could constitute a path towards tighter generalization bounds that depend on both the data-encoding and the trainable part of the PQC-based model.

In Section~\ref{s:results}, we will first prove generalization bounds for \smash{$\Fclass$}, which depend explicitly on properties of $\Omega$, before exploring in detail in Section~\ref{s:gen-bounds-pqc} how these relevant properties of $\Omega$ depend on the data-encoding strategy $\calD$.
Exploiting the fact that, for a given $B\geq\lVert M\rVert_\infty$, $\calF_{\Theta, \calD, M} \subseteq \calF_{\Omega(\calD)}^{B}$ then automatically yields explicitly encoding-dependent generalization bounds for $\calF_{\Theta, \calD, M}$.

As the connection between the data-encoding strategy $\calD$ and the set $\Omega(\calD)$ plays a crucial role, we illustrate this connection for a generic data-encoding strategy here. We first consider the action of a single encoding evolution $S(\xx)$ in the density matrix picture, where it acts via the quantum channel
\begin{align}
    \calS(\xx)[\rho] = \exp\left( -i \ee \xx H \right) \rho\exp\left( i \ee \xx H \right)  ,
\end{align}
where the Hamiltonian $H$ takes the role of any of the above Hamiltonian terms \smash{$H_j^{(i)}$} and $\ee$ can be any basis vector.
We can expand $\rho$ in the eigenbasis of the Hamiltonian $H | \lambda_k \rangle =  \lambda_k| \lambda_k \rangle$ and obtain
\begin{align}
    \calS(\xx)[\rho] &= \calS(\xx)\left[\sum_{k,l} \rho_{k,l} \, | \lambda_k \rangle\!\langle \lambda_l| \right] \\
    &= \sum_{k,l} \rho_{k,l}\, \calS(\xx)\left[| \lambda_k \rangle\!\langle \lambda_l| \right] \\
    &= \sum_{k,l} \rho_{k,l} \, \exp(-i (\lambda_k-\lambda_l) \ee  \xx)| \lambda_k \rangle\!\langle \lambda_l|.
\end{align}
We see that the differences of the eigenvalues $\lambda_k$ of the Hamiltonian $H$ determine the frequencies with which the different elements of the expansion of $\rho$ are multiplied. We can combine the different frequencies with the weight vector $\ee$ to obtain the set of all available frequencies
\begin{align}
    \Omega(H) = \{ \oomega_{k,l} = (\lambda_k-\lambda_l)\ee  \,|\, \lambda_k, \lambda_l \in \operatorname{spec}(H) \}.
\end{align}
With this notation, we can simplify our expression for $\calS(\xx)[\rho]$ to obtain
\begin{align}
    \calS(\xx)[\rho] &= \sum_{\oomega \in \Omega(H)} \exp(-i \oomega \xx) \rho_{\oomega},
\end{align}
where the operators $\rho_{\oomega}$ are given by collecting the terms in the above sum for which the frequency differences are the same, i.e.
\begin{align}
    \rho_{\oomega} = \sum_{(k,l) \in I(\oomega)} \rho_{k,l} | \lambda_k \rangle\!\langle \lambda_l|, \text{ where }I(\oomega) = \{ (k,l) \, | \, (\lambda_k-\lambda_l)\ee  = \oomega \}.
\end{align}
As $\rho$ is Hermitian, we have that $\rho_{\oomega} = \rho_{-\oomega}^{*}$.
The frequency structure carries over if we measure the expectation value of an arbitrary observable $M$ for the state $\calS(\xx)[\rho]$ to obtain a prediction
\begin{align}
    f(\xx) &= \Tr \{ \calS(\xx)[\rho] M \}
    = \sum_{\oomega \in \Omega(H)} \exp(-i \oomega \xx) \Tr \{ \rho_{\oomega} M \}
    = \sum_{\oomega \in \Omega(H)} {c_{\oomega}}\exp(-i \oomega \xx) .
\end{align}
As a result, we obtain a GTP with coefficients $c_{\oomega} = \Tr \{ \rho_{\oomega} M \}$. Note that, as $\rho_{\oomega} = \rho_{-\oomega}^{*}$, we have that $c_{\oomega} = c_{-\oomega}^{*}$, which ensures that $f(\xx)$ is real-valued as expected. The coefficients of this series could depend intricately on the circuit that was used to construct $\rho$ and on the specific observable $M$, but a profound understanding of this relation is an open question. However, this does not pose an obstacle for us, as only the set $\Omega$ is relevant for our study. 

We have just derived the frequency structure for one encoding gate, but for more complicated circuits we have to understand the action of multiple encoding gates, potentially interleaved with some trainable unitaries. The intermediary unitaries, however, will only result in a basis change, not affecting the set of combined frequencies. We can therefore ignore them and just consider the repeated action of two distinct encoding gates with Hamiltonians $H_1$ and $H_2$, resulting in
\begin{align}
    \calS_2(\xx)[\calS_1(\xx)[\rho]] &= \calS_2(\xx)\left[\sum_{{\oomega_1 \in \Omega(H_1)}} \exp(-i \oomega_1 \xx) \rho_{\oomega_1}\right] \\ 
    &= \sum_{\oomega_1 \in \Omega(H_1)} \exp(-i \oomega_1 \xx) \sum_{\oomega_2 \in \Omega(H_2)} \exp(-i \oomega_2 \xx)\rho_{\oomega_1,\oomega_2} \\
    &= \sum_{\oomega_1 \in \Omega(H_1)} \sum_{\oomega_2 \in \Omega(H_2)} \exp(-i [\oomega_1 + \oomega_2] \xx)\rho_{\oomega_1,\oomega_2}.
\end{align}
At this point, we precisely understand that the application of the second gate results in new frequencies that encompass all possible sums of the different frequencies. We can again consolidate this if we consider the sumset (or Minkowski sum) of the two sets of frequencies $\Omega(H_1)$ and $\Omega(H_2)$
defined as
\begin{align}
    \Omega(\{H_1, H_2\}) \coloneqq \Omega(H_1) + \Omega(H_2) \coloneqq \{ \oomega_1 + \oomega_2 \, | \, \oomega_1 \in \Omega(H_1), \oomega_2 \in \Omega(H_2) \}.
\end{align}
With that we have
\begin{align}
    \calS_2(\xx)[\calS_1(\xx)[\rho]] &= \sum_{{\oomega \in \Omega(H_1) + \Omega(H_2)}} \exp(-i \oomega \xx) \rho_{\oomega}.
\end{align}
Note again that the values of specific components $\rho_{\oomega}$ depend on the specific initial state $\rho$ and possible intermediate unitaries, but, in this work, we are only interested in $\Omega$ itself. We can apply the same logic recursively to see that the set of accessible frequencies for any encoding strategy $\calD$ is given by the sumset of all the individual sets of frequencies \smash{$\Omega(H^{(i)}_j)$} for each gate:
\begin{align}\label{eqn:omega_full_definition}
    \Omega(\calD) = \sum_{\calD^{(i)} \in \calD} \sum_{H \in \calD^{(i)}} \Omega(H)
    = \sum_{i=1}^{d} \sum_{j=1}^{N^{(i)}} \{  (\lambda_k-\lambda_l)\ee^{(i)} \,|\, \lambda_k, \lambda_l \in \operatorname{spec}(H_j^{(i)}) \}.
\end{align}

\section{Prior and related work}\label{s:prior_work}

Before presenting our explicitly encoding-dependent generalization bounds for PQC-based models in the next two sections, we discuss how our results compare to prior work. While there is a massive amount of prior and ongoing work on the generalization capacity of classical models, see for example the survey in Ref.~\cite{jiang}, such results have only recently begun to emerge for PQC-based models. Here, we focus on a comparison with these latter results. Additionally, while the following paragraphs constitute a detailed review of existing generalization bounds for PQC-based models, we stress that no knowledge of these prior works is necessary to understand our proofs and results. In particular, the presentation here is intended to establish context for our work and to place prior works in relation to each other, but the remainder of this manuscript can safely be read independently of the review presented here.

Given the discussions in the previous two sections, we note that, at a high level, all prior work on generalization bounds for PQC-based models can be classified via the following three criteria:
\begin{enumerate}
    \item Which restrictions -- if any -- are placed on the architecture/structure of the PQCs generating the model class considered?
    \item In terms of which architectural hyper-parameters, or experimentally accessible quantities, are the generalization bounds expressed?
    \item Via which complexity measure are the generalization bounds derived?
\end{enumerate}
Given this, we will use the above questions as guidelines for understanding and relating existing results. 
Throughout this discussion, keep in mind that, as explained in Section~\ref{s:intro}, all prior works are restricted to \textit{encoding-first} models, whereas we allow for data re-uploading.

Additionally, while some of the following works study the same complexity measures as the ones examined here -- namely, Rademacher complexity and covering numbers -- all of them differ from ours in both the restriction to encoding-first PQC-based models and in a lack of \textit{explicit} dependence on the data-encoding strategy. Given this, we split our survey into two parts. 
First, in Section~\ref{ss:encoding_independent}, we discuss those prior works which derive \textit{encoding-independent} generalization bounds.
In Section~\ref{ss:encoding_dependent}, we then discuss existing works deriving generalization bounds which depend on the data-encoding strategy, but with a dependence which is implicit, and not necessarily clear a priori.

\subsection{Encoding-independent complexity and generalization bounds}\label{ss:encoding_independent}

Ref.~\cite{CaroDatta2020} is an early study
of the complexity and generalization capacity of quantum circuit based models, which presents encoding-independent bounds on the pseudo-dimension of function classes associated with encoding-first $2$-local (unitary or CPTP) PQCs, polynomial in the size (number of gates) and depth of the trainable part of the circuit (in which all gates were considered trainable). Such pseudo-dimension bounds then yield generalization bounds, which also depend polynomially on the size and depth of the trainable circuit.
Ref.~\cite{popescu2021learning} has extended the generalization bounds of Ref.~\cite{CaroDatta2020} to the agnostic setting.
In a similar vein, Ref.~\cite{du2021efficient} has recently derived encoding-independent covering number bounds for encoding-first PQC-based models, which depend explicitly on the number of gates in the PQC, and the operator norm of the measured observable. Once again, using standard tools from statistical learning theory, the authors of Ref.~\cite{du2021efficient} are then able to use these covering number bounds to provide an encoding-independent generalization bound.

Working from the perspective of kernel methods, Ref.~\cite{Vedran2021} has recently investigated the complexity of encoding-first PQC-based models in terms of properties of the parametrized measurement which follows data-encoding. More specifically, they interpret the entire parametrized circuit following the data-encoding as a parametrized measurement, and provide bounds for
the VC-dimension of the model class in terms of the rank of the parametrized observable, and for
the fat-shattering dimension in terms of the Frobenius norm of the parametrized observable. These bounds on standard complexity measures then allow them to prove generalization bounds which depend explicitly on either the rank or the Frobenius norm of the accessible observables. 
However, similarly the perspective we advocate in this work, the authors of Ref.~\cite{Vedran2021} stress the application of generalization bounds for model selection, via structural risk minimization.

Finally, Ref.~\cite{Bu.2021b} has recently initiated a resource-theoretic approach by providing encoding-independent bounds on both the Rademacher and Gaussian complexity of encoding-first PQC-based models, in terms of the number of repetitions of resource channels allowed in the PQC. 
These Rademacher and Gaussian complexity bounds have then been used to derive generalization bounds, which depend on the same quantities, and therefore provide an encoding-independent 
resource-theoretic perspective on generalization in encoding-first 
PQC-based models. 

\subsection{Encoding-dependent complexity and generalization bounds}\label{ss:encoding_dependent}

We proceed by discussing prior work deriving generalization bounds which do depend on the data-encoding strategy.
While the dependence on the data-encoding could take various forms,
in this manuscript
we aim to derive generalization bounds which depend \textit{explicitly} on architectural hyper-parameters related to the data-encoding strategy (such as the number of encoding gates of a specific type), and therefore facilitate the straightforward implementation of model selection via structural risk minimization. 
This is in contrast to all of the prior encoding-dependent generalization bounds, which are written in terms of some quantity which depends on the data-encoding strategy, but with an implicit dependence which is not 
a priori
clear, and needs to be assessed experimentally. Given this fundamental difference between our 
generalization bounds and those of the
prior works we discuss here, a 
natural
open
question is whether the implicitly encoding-dependent quantities used in the following works can be written explicitly in terms of architectural hyper-parameters related to the data-encoding strategy. If possible, this would immediately provide explicitly encoding-dependent generalization bounds comparable to those we derive in this work.

With this in mind, we begin our survey of implicitly encoding-dependent generalization bounds with Ref.~\cite{abbas2020power}, which has  suggested a complexity measure based on the classical Fisher information, called the effective dimension, and demonstrated that one can indeed state generalization bounds in terms of the effective dimension. Utilizing the empirical Fisher information as a tool for approximating the effective dimension, Ref.~\cite{abbas2020power} presented numerical experiments which demonstrate a clear dependence of the effective dimension on the encoding-strategy.
However, the explicit dependence of the effective dimension on the encoding strategy is not clear and needs to be evaluated experimentally. Additionally, Ref.~\cite{abbas2020power} also provided a comparison between the effective dimension of PQC-based models and comparable classical models, and demonstrated that PQC-based models can exhibit a higher effective dimension. While not discussed explicitly in Ref.~\cite{abbas2020power}, we stress, however, that one should \textit{not} use model complexity (e.g., effective dimension) as the sole criterion for model selection, since model classes with higher effective dimension may have worse generalization behavior than models with a lower effective dimension. Instead, as we advocate in this work, one should ideally use a framework such as structural risk minimization to select a model with the smallest upper bound on out-of-sample performance.

Also working from an information theoretic perspective, and with a
focus on the role of data-encoding, Ref.~\cite{banchi2021generalization} has recently presented generalization bounds for PQC-based models in terms of information-theoretic quantities describing a notion of mutual information between the post-encoding quantum state $\rho(\xx)$ and the classical data. While these generalization bounds have a strong implicit dependence on the data-encoding strategy, 
it is once again not immediately clear, apart from in a few special cases, how to explicitly express the suggested complexity measure in terms of architectural hyper-parameters related to the data-encoding strategy. 

From a resource theoretic perspective, and complementing Ref.~\cite{Bu.2021b}, the series of works~\cite{Bu.2021a, Bu.2021c} have further studied the Rademacher complexity of encoding-first PQC-based models. However, unlike in Ref.~\cite{Bu.2021b}, the Rademacher complexity bounds of Refs.~\cite{Bu.2021a, Bu.2021c}  are given in terms of quantities that exhibit an implicit dependence on the data-encoding strategy. More specifically, Ref.~\cite{Bu.2021c} provides Rademacher complexity bounds in terms of the size, depth and amount of magic available as a resource. Additionally, Ref.~\cite{Bu.2021a} also studies noisy PQC-based models and provides Rademacher complexity bounds in terms of either the Rademacher complexity of the associated noiseless circuit or the free-robustness of the model. 

Recently, Ref.~\cite{chen2021expressibility} has studied generalization for PQC-based models using a hardware efficient ansatz with a specific choice of data-encoding. 
For this setting, they proved VC-dimension bounds that scale polynomially with the minimum of the number of qubits and the number of trainable layers. In their proofs, they combine light cone arguments with a trigonometric function representation for functions implemented by their ansatz.

Finally, we mention Ref.~\cite{huang2021power} which has developed techniques for evaluating the potential advantages of quantum kernels over classical kernels. These results are of relevance to this work due to the close relationship between PQC-based models and kernel methods~\cite{schuld2021kernelmethods}.
In a first step, the authors of Ref.~\cite{huang2021power}  suggest the evaluation of a geometric quantity which depends on the chosen quantum feature map and the available training data instances. If the quantum machine learning model passes this first test, a model complexity parameter, which now depends on the quantum encoding and the training data (both instances and labels), should be computed.
While these complexity measures can be classically computed in time polynomial in the training data size, analytically determining their exact dependence on the data-encoding can be challenging. This is in contrast to our model complexity bounds, which depend straightforwardly on hyper-parameters associated with the data-encoding strategy, such as the number of encoding gates of a specific type.


\section{Generalization bounds for generalized trigonometric polynomials}\label{s:results}

We recall (from Section~\ref{s:pqc_models}) that we can prove generalization bounds on $\mathcal{F}_{\Theta,\mathcal{D},M}$, the 
hypothesis
class of interest for a given PQC-based model, by proving generalization bounds on $\Fclass$. Recall that $\Fclass$ has been defined as the class of generalized trigonometric polynomials (GTPs) with frequencies in $\Omega$ and infinity-norm bounded by $B$ as
\begin{align}\label{e:fclass_complex}
    \Fclass 
    =\Bigg\{[0,2\pi)^d &\ni \xx\mapsto f(\xx) = \sum\limits_{\oomega\in\Omega}c_{\oomega} \exp(-i \oomega \xx) \, \Big\vert \, (c_{\oomega})_{\oomega\in\Omega} \text{ such that } \norm{f}_{\infty} \leq B\Bigg\}.
\end{align}
In order to prove generalization bounds for $\Fclass$, it will be convenient to work with the cosine and sine representation of the complex exponential, and with the norm of the vector of coefficients instead of the norm of the function. 
Note that, since we have observed in Section~\ref{s:pqc_models} that $c_{-\oomega}=c_{\oomega}^\ast$, we can define, for every $\oomega\in\Omega$
\begin{align}
    a_{\oomega} & \coloneqq c_{\oomega} + c_{-\oomega} \in\R ,\\
    b_{\oomega} & \coloneqq \frac{1}{i}(c_{\oomega} - c_{-\oomega})\in\R.
\end{align}
With these, it further follows that
\begin{align}
    c_{\oomega}e^{-i\oomega\xx} + c_{-\oomega}e^{i\oomega\xx} & = a_{\oomega}\cos(\oomega\xx) + b_{\oomega}\sin(\oomega\xx),
\end{align}
which allows us to rewrite the sum in Eq.~\eqref{e:fclass_complex} as a sum of real terms only.
If we were only considering frequencies given by real numbers, then it would suffice to sum over the non-negative frequencies in the real sum representation.
However, we are dealing with frequency 
vectors. As this is the case, we start by removing the zero vector from the set of frequencies to obtain $\Omega^\ast\coloneqq\Omega\setminus\{0\}$. 
Note that this is meaningful as $0\in\Omega$ for any $\Omega$ of the form introduced in Section~\ref{s:pqc_models}.
Next, we divide $\Omega^\ast$ into two disjoint parts $\Omega^\ast=\Omega_+ \cup\Omega_-$, with $\Omega_+\cap\Omega_- =\emptyset$, such that for every $\oomega\in\Omega_+$ we have that  $-\oomega\in\Omega_-$. We again note that this is possible due to the specific form of the sets $\Omega$ discussed in Section~\ref{s:pqc_models}. In particular, we then have $|\Omega|=2|\Omega_+|+1$. Additionally, we make use of a shorthand notation for the vectors $(a_{\oomega})_{\oomega\in\Omega_+}$ and $(b_{\oomega})_{\oomega\in\Omega_+}$: We keep the indices outside of the parentheses, but remove the indexing set. Namely we write $(a_0,(a_{\oomega})_{\oomega},(b_{\oomega})_{\oomega})$ in place of $(a_0,(a_{\oomega})_{\oomega\in\Omega_+},(b_{\oomega})_{\oomega\in\Omega_+})$.
We only explicitly write the indexing set at certain points to avoid confusion.

With these notational points in mind, we can rewrite the hypothesis class $\Fclass$ as
\begin{align}
    \begin{split}
    \Fclass
    = \Bigg\{[0,2\pi)^d&\ni \xx\mapsto f(x) = \frac{a_0}{2} + \sum\limits_{\oomega\in\Omega_+} \left(a_{\oomega}\cos(\oomega\xx) + b_{\oomega}\sin(\oomega\xx)\right)\\ 
    &\Bigg\vert~ (a_0,(a_{\oomega})_{\oomega},(b_{\oomega})_{\oomega}) \text{ such that } \norm{f}_{\infty} \leq B \Bigg\},
    \end{split}
\end{align}
and, for $\tilde{B}>0$,
we define the class $\Fclasstrignew$ via
\begin{align}\label{eq:trig_polys_class}
\begin{split}
    \Fclasstrignew
    \coloneqq \Bigg\{[0,2\pi)^d&\ni \xx\mapsto \frac{a_0}{2} + \sum\limits_{\oomega\in\Omega_+} \left(a_{\oomega}\cos(\oomega\xx) + b_{\oomega}\sin(\oomega\xx)\right)\hspace{1.25cm}\\
    &\Bigg\vert~ \norm{(a_0,(a_{\oomega})_{\oomega},(b_{\oomega})_{\oomega})}_2 \leq \tilde{B}\Bigg\},
\end{split}
\end{align}
where the $2$-norm is given by
\begin{align}
    \norm{(a_0,(a_{\oomega})_{\oomega},(b_{\oomega})_{\oomega})}_2 
    \coloneqq \sqrt{a_0^2 + \sum_{\oomega\in\Omega_+} (a_{\oomega}^2 + b_{\oomega}^2)}.
\end{align}
We assume for the purposes of this section that $\tilde{B}>0$ is chosen such that $\Fclass\subseteq\Fclasstrignew$ holds true.
In the special case of frequency vectors with integer entries, i.e., $\Omega\subseteq\mathbb{Z}^{d}$, we can always choose $\tilde{B}=2B$. This can be seen as follows: For a function $f\in\Fclass$ given by $f(\xx)=\sum_{\oomega\in\Omega}\exp(-i \oomega \xx) c_{\oomega}={a_0}/{2}+\sum_{\oomega\in\Omega_+} \left(a_{\oomega}\cos(\oomega\xx) + b_{\oomega}\sin(\oomega\xx)\right)$, we obtain
\begin{align}\label{e:gtp_coeff_bound}
    \norm{(a_0,(a_{\oomega})_{\oomega\in\Omega_+},(b_{\oomega})_{\oomega\in\Omega_+})}_2 
    &= 2\norm{(c_0,(c_{\oomega})_{\oomega\in\Omega^\ast})}_2
    = \frac{2}{(2\pi)^{\nicefrac{d}{2}}}\norm{f}_2
    \leq 2\norm{f}_\infty
    \leq 2B ,
\end{align}
so we can take $\tilde{B}=2B$ to ensure $\Fclass\subseteq\Fclasstrignew$. Here, we used that $\Omega\subseteq\mathbb{Z}^{d}$ implies orthogonality of the functions $\xx\mapsto \exp(-i \oomega \xx)$, $\oomega\in\Omega$, which in turn implies $\norm{f}_2=(2\pi)^{\nicefrac{d}{2}}\norm{(c_0,(c_{\oomega})_{\oomega\in\Omega^\ast})}_2$ since each of those functions has $2$-norm equal to $(2\pi)^{\nicefrac{d}{2}}$.
As a consequence of the assumption that $\Fclass\subseteq\Fclasstrignew$, generalization bounds uniform over $\Fclasstrignew$ imply generalization bounds uniform over $\Fclass$. Therefore, we focus on proving generalization bounds for $\Fclasstrignew$.

We conjecture that the inclusion $\calF_{\Theta, \calD, M} \subseteq \Fclass$ observed in Section \ref{s:pqc_models} can be strengthened:
\begin{conjecture}\label{conjecture:strengthened-inclusion}
    For any PQC-based model, we have
    \begin{align}
    \begin{split}
        \calF_{\Theta, \calD, M} 
        \subseteq 
        \Fclasstil 
        \coloneqq \Bigg\{[0,2\pi)^d &\ni \xx\mapsto f(\xx) = \sum\limits_{\oomega\in\Omega} c_{\oomega}\exp(-i \oomega \xx) \, \\&\Big\vert \, (c_{\oomega})_{\oomega\in \Omega} \text{ s.t. } \norm{f}_{\infty}\leq B \text{ and } \norm{(c_{\oomega})_{\oomega\in \Omega}}_\infty\leq B\Bigg\} .
    \end{split}
\end{align}
\end{conjecture}
If this conjecture is correct, then it implies that $\calF_{\Theta, \calD, M} \subseteq \Fclasstil\subseteq \Fclasstrignew$ holds with $\tilde{B}= 2\sqrt{\lvert\Omega\rvert}B$ and $\norm{M}_\infty\leq B$ for general frequency spectra $\Omega=\Omega(\mathcal{D})$, because we have the inequalities $\norm{(a_0,(a_{\oomega})_{\oomega},(b_{\oomega})_{\oomega})}_2 \leq \sqrt{\lvert\Omega\rvert} \norm{(a_0,(a_{\oomega})_{\oomega},(b_{\oomega})_{\oomega})}_\infty\leq 2\sqrt{\lvert\Omega\rvert} \norm{(c_{\oomega})_{\oomega\in \Omega}}_\infty$. Thus, conditional on Conjecture~\ref{conjecture:strengthened-inclusion}, generalization bounds for $\Fclasstrignew$ with suitable $\tilde{B}=2\sqrt{\lvert\Omega\rvert}\norm{M}_\infty$ imply generalization bounds for PQCs with general encoding strategies and potentially non-integer frequency spectra.

Our bounds focus on the dependence of generalization on the frequency spectrum $\Omega$.
We obtain these bounds from bounds on the complexity of $\Fclasstrignew$, measured in terms of two complexity measures from classical learning theory, namely the \emph{Rademacher complexity} and the \emph{metric entropy}. We first recall the definitions of these important quantities and then give an overview over our results and proof strategy.

\begin{definition}[(Empirical) Rademacher complexity]\label{def:rademacher-complexities}
Let $\calZ$ be some data space, $\mathcal{F}\subseteq\mathbb{R}^\mathcal{Z}$ a function class, and $S=(\zz_1,\ldots,\zz_m)\in\mathcal{Z}^m$. The \emph{empirical Rademacher complexity} of $\mathcal{F}$ with respect to~$S$ is defined as
\begin{align}
\hat{\mathcal{R}}_{S}(\mathcal{F})\coloneqq\underset{\ssigma\sim U(\lbrace -1,1\rbrace^m)}{\mathbb{E}}\Big[\sup\limits_{f\in\mathcal{F}}\frac{1}{m}\sum\limits_{i=1}^m \sigma_i f(\zz_i)\Big],
\end{align}
where $U(\lbrace -1,1\rbrace^m)$ denotes the uniform distribution on $\lbrace -1,1\rbrace^m$. The i.i.d.~random variables $\sigma_1,\ldots,\sigma_m$ are often called \emph{Rademacher random variables}.
\end{definition}

For later use, we note that, if $\mathcal{F}\subseteq\mathcal{G}\subseteq\mathbb{R}^\mathcal{Z}$, then, for any $S\in\mathcal{Z}^m$ we have $\hat{\mathcal{R}}_{S}(\mathcal{F})\leq \hat{\mathcal{R}}_{S}(\mathcal{G})$. Next, we introduce our second complexity measure:

\begin{definition}[Covering nets, covering number, and metric entropy]\label{def:covering-numbers}
Let $(X,d)$ be a (pseudo-)metric space. Let $K\subseteq X$ and let $\varepsilon >0$. We call $N\subseteq K$ an (interior) $\varepsilon$-covering net of $K$ if for all $x \in K$ there exists a $y \in N$ such that $d(x,y)\leq \varepsilon$. 
The \emph{covering number} $\mathcal{N}(K,d,\varepsilon)$ is defined as the smallest possible cardinality of an (interior) $\varepsilon$-covering net of $K$. Finally, we define the \emph{metric entropy} $\log_2\mathcal{N}(K,d,\varepsilon)$ via a logarithm of the covering number.
\end{definition}

For our purposes, the relevant covering numbers are those of $\Fclasstrignew$ with respect to the pseudo-metrics induced by the data-dependent semi-norms $\norm{\cdot}_{2,S\rvert_x}$, which, given training data $S=\{(\xx_i,y_i)\}_{i=1}^m$, are defined as 
\begin{equation}
\norm{f}_{2,S\rvert_x}\coloneqq \sqrt{
\frac{1}{m}\sum_{i=1}^m \, \lvert f(\xx_i)\rvert^2
}.
\end{equation}

In Section~\ref{sbs:gen-bounds-trig-poly-rademacher}, we prove Rademacher complexity bounds for $\Fclasstrignew$. We do so by understanding $\Fclasstrignew$ as (a subset of) a class of functions implemented by a simple classical \emph{neural network} (NN) with a single hidden layer and with sinusoidal activation functions in the hidden layer. For such NN architectures, we can then apply already known Rademacher complexity bounds. This strategy leads to 
\begin{equation}
    \hat{\mathcal{R}}_{S\rvert_x}(\Fclass)
    \leq \hat{\mathcal{R}}_{S\rvert_x}(\Fclasstrignew) 
    \leq \tilde{\mathcal{O}}\left(\sqrt{\frac{\lvert\Omega\rvert}{m}}\right)
\end{equation}
for a training data set $S$ of size $m$, with data instances $S\rvert_x=\{x_i\}_{i=1}^m$. Here, the $\tilde{\mathcal{O}}$ refers to the asymptotic behavior as $\lvert\Omega\rvert, m\to\infty$ and hides a logarithmic dependence on $\lvert\Omega\rvert$. (As we are most interested in the dependence on $\lvert\Omega\rvert$, we also hide the dependence on $B$ and $\tilde{B}$ here.)
With these Rademacher complexity bounds at hand, we can then derive generalization guarantees for $\Fclasstrignew$, and thus $\Fclass$, using a standard generalization bound in terms of the Rademacher complexity. We obtain that for a bounded Lipschitz loss function, with probability $\geq 1-\delta$, the generalization error satisfies
\begin{equation}
    R(f)-\hat{R}_S(f) 
    \leq \tilde{\mathcal{O}}\left(\sqrt{\frac{\lvert\Omega\rvert}{m}} + \sqrt{\frac{\log(\nicefrac{1}{\delta})}{m}}\right),
\end{equation}
uniformly over $f\in\Fclasstrignew$ for training data $S$ of size $m$. Again, we emphasize the leading-order dependence on $\lvert\Omega\rvert$ and hide other parameters. We note that, without further assumptions, as in classical agnostic learning scenarios, we do not expect a better scaling with respect to $m$ than the Hoeffding-like $\sim\nicefrac{1}{\sqrt{m}}$. 

In Section~\ref{sbs:gen-bounds-trig-poly-covering}, we bound the covering number and metric entropy of $\Fclasstrignew$,
and thus of $\Fclass$.
We achieve this by constructing a covering net for $\Fclasstrignew$ from a suitable (finer-grained) covering net of the allowed vectors of Fourier coefficients. Here, we crucially use that $\lvert\Omega\rvert$ determines the dimension of the space in which we have to take these covering nets. With this reasoning, we obtain a metric entropy bound of
\begin{equation}
    \log_2\mathcal{N}(\Fclass,\norm{\cdot}_\infty,\varepsilon)
    \leq \log_2\mathcal{N}(\Fclasstrignew,\norm{\cdot}_\infty,\frac{\varepsilon}{2})
    \leq \tilde{\mathcal{O}}\left(\lvert\Omega\rvert \log(\nicefrac{1}{\varepsilon}) \right),
\end{equation}
where the $\tilde{\mathcal{O}}$ hides logarithmic dependencies on $B$ and $\lvert\Omega\rvert$. 
Given these metric entropy bounds, we then use the chaining method to derive empirical Rademacher complexity bounds. Again assuming a bounded Lipschitz loss function, this method yields, with probability $\geq 1-\delta$, a generalization error bound of 
\begin{equation}
    R(f)-\hat{R}_S(f) \leq \tilde{\mathcal{O}}\left(\sqrt{\frac{\lvert\Omega\rvert}{m}} + \sqrt{\frac{\log(\nicefrac{1}{\delta})}{m}}\right),
\end{equation}
simultaneously for all $f\in\Fclass\subseteq\Fclasstrignew$, assuming training data of size $m$ and hiding  both logarithmic terms in $\lvert\Omega\rvert$ or $\tilde{B}$ as well as dependencies on $B$, the Lipschitz constant, and the bound on the loss.
While we see that, with the above definition of 
$\Fclass$ and
$\Fclasstrignew$, the strategies of Sections~\ref{sbs:gen-bounds-trig-poly-rademacher} and \ref{sbs:gen-bounds-trig-poly-covering} lead to the same generalization bound in leading order, we nevertheless present both approaches because they yield different results if the assumption on the Fourier coefficients appearing in 
$\Fclass$ or
$\Fclasstrignew$ is changed from a $2$-norm bound to a general $p$-norm bound.

In the light of the discussion in Section~\ref{s:pqc_models}, these generalization bounds for classes of generalized trigonometric polynomials imply generalization bounds for PQCs. As we have focused on the dependence on the frequency spectrum in the former, we obtain a focus on the encoding-dependence in the latter. We provide and discuss these results in Section~\ref{s:gen-bounds-pqc}.

\subsection{Generalization bounds for generalized trigonometric polynomials via Rademacher complexity}\label{sbs:gen-bounds-trig-poly-rademacher}

We begin our analysis by stating our Rademacher complexity bound for $\Fclasstrignew$. As we will see, this bound is obtained by combining two partial results, and will lead directly to a generalization bound. For ease of notation, we write $K_i\coloneqq\max_{\oomega\in\Omega_+}\{\lvert \omega_i\rvert\}$ for $i\in\{1,\ldots,d\}$ and $K\coloneqq\sum_i K_i$.

\begin{lemma}[Rademacher complexity bounds for GTPs]\label{lemma:rademacher-complexity-trig-polys}
Let $d,m\in\mathbb{N}$. Let $S\rvert_x\in([0,2\pi)^d)^m$. Let $\Fclasstrignew$ be as defined in Eq.~\eqref{eq:trig_polys_class}. The empirical Rademacher complexity of $\Fclasstrignew$ with respect to~$S\rvert_x\coloneqq(\xx_1,\ldots,\xx_m)$ can be upper-bounded as
\begin{align}
    \hat{\mathcal{R}}_{S\rvert_x} (\Fclasstrignew) 
    &\leq \mathcal{O}\left(\frac{\min\left\{\sqrt{\log(2d)}\max\{K,\tilde{B}\sqrt{\lvert\Omega\rvert} \}, \tilde{B}\sqrt{\lvert\Omega\rvert\log(\lvert\Omega\rvert)}\right\}}{\sqrt{m}} \right).
\end{align}
\end{lemma}

In order to prove Lemma~\ref{lemma:rademacher-complexity-trig-polys} we state and show two partial results, namely Lemmas~\ref{lemma:Old_V1} and~\ref{lemma:Old_V2}.
These two Lemmata have slightly different proof strategies, but both are motivated by thinking of generalized trigonometric polynomials as being realized by certain neural network architectures.

\begin{lemma}[Empirical Rademacher complexity of $\Fclasstrignew$---Version $1$]\label{lemma:Old_V1}
    Let $d, m$, $S\rvert_x$, and $\Fclasstrignew$ be as in Lemma~\ref{lemma:rademacher-complexity-trig-polys}.
    Then, the empirical Rademacher complexity of $\Fclasstrignew$ with respect to~$S\rvert_x$ can be upper-bounded as
    \begin{align}
        \empRad_{S\rvert_x}(\Fclasstrignew) 
        & \leq \calO\left(\frac{1}{\sqrt{m}}\max\{K,\tilde{B}\sqrt{\lvert\Omega\rvert}\}\sqrt{\log(2d)}\right).
    \end{align}
\end{lemma}

\begin{proof}
    We prove this statement by constructing a function class that contains $\Fclasstrignew$ and whose empirical Rademacher complexity we are able to upper bound by viewing it as arising from a simple layered neural network (NN) architecture. More specifically, we consider the following class of functions 
    \begin{align}\label{eq:lemma1_V1_aux}
    \begin{split}
        \calG^{\tilde{B}}_\Omega 
        \coloneqq \Bigg\{[0,2\pi)^d&\ni\xx\mapsto \frac{d_0}{2} + \sum_{\oomega\in\Omega_+} d_{\oomega}\sin(\aalpha_{\oomega}\xx + \gamma_{\oomega})\\
        &\Bigg|\,
        \norm{(d_0,(d_{\oomega})_{\oomega})}_2\leq \tilde{B},\, \aalpha_{\oomega}\in\prod\limits_{i=1}^d [-K_i,K_i], \, \gamma_{\oomega}\in[-\pi,\pi ) \Bigg\},
    \end{split}
    \end{align}
    which can be realized by a NN with a single hidden layer of neurons with sine activation functions, and a linear activation at the output neuron.
    Here, again $(d_{\oomega})_{\oomega}$ stands for the vector $(d_{\oomega})_{\oomega\in\Omega_+}$.
    Also, note that for every $\oomega\in\Omega_+$, $\aalpha_{\oomega}$ is a $d$-dimensional vector and $\gamma_{\oomega}$ a real number.
    
    We claim that $\Fclasstrignew\subseteq\calG^{\tilde{B}}_\Omega$.
    We can prove this inclusion directly by finding the corresponding parameters $(d_0, (d_{\oomega})_{\oomega})$, $(\gamma_{\oomega})_{\oomega}$ and $(\aalpha_{\oomega})_{\oomega}$ for each element $f\in\Fclasstrignew$, specified by the corresponding $(a_0, (a_{\oomega})_{\oomega},(b_{\oomega})_{\oomega})$.
    We can find a valid assignment term by term.
    We start by noting $d_0=a_0$.
    Next, we spell out the term corresponding to the frequency vector $\oomega$ with the well-known angle sum trigonometric 
    identity
    \begin{align}
        d_{\oomega}\sin(\aalpha_{\oomega}\xx + \gamma_{\oomega}) 
        & = d_{\oomega}\cos(\gamma_{\oomega})\sin(\aalpha_{\oomega}\xx) + d_{\oomega}\sin(\gamma_{\oomega})\cos(\aalpha_{\oomega}\xx).
    \end{align}
    Now, for any given $(a_{\oomega})_{\oomega}$ and $(b_{\oomega})_{\oomega}$, we can set
    \begin{align}
        d_{\oomega} & \coloneqq \sqrt{a_{\oomega}^2 + b_{\oomega}^2},~
        \aalpha_{\oomega} \coloneqq \oomega, \text{ and }
        \gamma_{\oomega} \coloneqq \arctan(b_{\oomega}/a_{\oomega}).
    \end{align}
    At this point, it is important to confirm that the assignment is valid within the restrictions imposed in Eq.~\eqref{eq:lemma1_V1_aux}.
    To begin with, we note that the $2$-norm bound from Eq.~\eqref{e:gtp_coeff_bound}, i.e. $\norm{(a_0, (a_{\oomega})_{\oomega}, (b_{\oomega})_{\oomega})}_2\leq \tilde{B}$, translates directly into $\norm{(d_0,(d_{\oomega})_{\oomega})}_2\leq \tilde{B}$, since $d_{\oomega}^2 = a_{\oomega}^2 + b_{\oomega}^2$ for all $\oomega$.
   Additionally, one can also see that the components of $\aalpha_{\oomega}$ are nothing but the frequencies  $\omega_i$ for each data coordinate, which fall in the interval $[-K_i,K_i]$ by construction.
    Finally, as a function $\arctan$ can output any angle, choosing the branch $[-\pi,\pi )$ is valid.
    With these, we reach
    \begin{align}
        d_{\oomega}\sin(\aalpha_{\oomega}\xx + \gamma_{\oomega}) 
        & = a_{\oomega}\sin(\oomega\xx) + b_{\oomega}\cos(\oomega\xx),
    \end{align}
    which has been our goal.
    
    As $\calG^{\tilde{B}}_\Omega$ arises from a NN whose activation functions are $1$-Lipschitz, continuous and anti-symmetric, we can use Lemma~\ref{cor:2.11Wolf} (stated in Appendix~\ref{appendix:auxiliary-results}).
    For that, we require upper bounds for the $1$-norm of the weight vector going into each neuron and for the moduli of the biases.
    For every neuron in the hidden layer, there are $d$ incoming weights, one for each data dimension, corresponding to the $d$ input neurons.
    Each component of those weight vectors ($\aalpha_{\oomega}$ in Eq.~\eqref{eq:lemma1_V1_aux}) takes values in $\in[-K_i,K_i]$ for some $i\in\{1,\ldots,d\}$, so the $1$-norm of such a weight vector is upper bounded by $K$.
    
    At the output neuron, there are $\lvert\Omega_+\rvert$ incoming weights ($d_{\oomega}$ in Eq.~\eqref{eq:lemma1_V1_aux}) and we have a bound on the $2$-norm of this weight vector.
    Therefore, Hölder's inequality applied to the $2$-norm gives the $1$-norm bound
    \begin{equation}
    \norm{(d_{\oomega})_{\oomega}}_1
    \leq \tilde{B}\sqrt{\lvert\Omega_+\rvert}.
    \end{equation}
    With that, we now know that the $1$-norm of any weight vector in the NN is upper bounded by $\max\{K,\tilde{B}\sqrt{\lvert\Omega_+\rvert}\}$.
    
    Next, we note that the modulus of the biases is at most $\pi$ in the hidden layer, and  $\tilde{B}$ in the output layer. As a result, we have that the moduli of the biases in the NN are upper bounded by $\max\{\pi,\tilde{B}\}$.
    Now that we have collected all the ingredients, we can plug them into Lemma~\ref{cor:2.11Wolf} and obtain the bound
    \begin{align}
        \empRad_{S\rvert_x}(\calG^{\tilde{B}}_\Omega) 
        & \leq \frac{1}{\sqrt{m}}\left(2\pi \max\{K,\tilde{B}\sqrt{\lvert\Omega_+\rvert}\}\sqrt{2\log(2d)} + \max\{\pi,\tilde{B}\}\right) \\
        & \leq \calO\left(\frac{1}{\sqrt{m}}\max\{K,\tilde{B}\sqrt{\lvert\Omega\rvert}\}\sqrt{\log(2d)}\right), 
    \end{align}
    where the $\calO$ notation refers to the scaling in $|\Omega|$.
    As $\calG^{\tilde{B}}_\Omega$ contains $\Fclasstrignew$ as a subset, this bound directly implies
    \begin{align}
        \empRad_{S\rvert_x}(\Fclasstrignew)
        \leq \calO\left(\frac{1}{\sqrt{m}}\max\{K,\tilde{B}\sqrt{\lvert\Omega\rvert}\}\sqrt{\log(2d)}\right),
    \end{align}
    which completes the proof.
\end{proof}

In the proof of Lemma \ref{lemma:Old_V1}, we do not bound the empirical Rademacher complexity of $\Fclasstrignew$ directly, rather we embed it into a larger class $\calG^{\tilde{B}}_\Omega$ whose complexity we then bound.
However, whereas only a discrete set of frequencies is used in $\Fclasstrignew$, the class $\calG^{\tilde{B}}_\Omega$ allows for a continuum of frequencies. In Lemma \ref{lemma:Old_V2}, we modify the idea of the previous proof to avoid this overcounting of frequencies.

\begin{lemma}[Empirical Rademacher complexity of $\Fclasstrignew$---Version $2$]\label{lemma:Old_V2}
    Let $d, m$, $S\rvert_x$, and $\Fclasstrignew$ be as in Lemma~\ref{lemma:rademacher-complexity-trig-polys}.
    Then, the empirical Rademacher complexity of $\Fclasstrignew$ with respect to~$S\rvert_x$ can be upper-bounded as
    \begin{align}
        \empRad_{S\rvert_x}(\Fclasstrignew) 
        & \leq \mathcal{O}\left(\frac{\tilde{B}}{\sqrt{m}}\sqrt{\lvert\Omega\rvert\log(\lvert\Omega\rvert)}\right).
    \end{align}
\end{lemma}

\begin{proof}

    Analogously to the proof of Lemma~\ref{lemma:Old_V1}, we provide an empirical Rademacher complexity upper bound for a larger function class $\Tilde{\calH}_\Omega^{\tilde{B}}$.
    Along the way, we see that the inclusion $\Fclasstrignew\subseteq\Tilde{\calH}_\Omega^{\tilde{B}}$ holds, so that the uniform bound we derive for the larger set is immediately inherited for the smaller one. We start by defining an auxiliary set of functions: let $\calM_{\Omega}$ be the set of generalized trigonometric monomials over $\R^d$ with frequency values in $\Omega_+$, defined as 
    \begin{align}
        \calM_{\Omega}
        & \coloneqq \{0\}\cup\left\{[0,2\pi)^d\ni\xx\mapsto\cos(\oomega\xx) \,|\, \oomega\in\Omega_+\right\} \cup \left\{[0,2\pi)^d\ni\xx\mapsto\sin(\oomega\xx) \,|\, \oomega\in\Omega_+\right\}.
    \end{align}
    Now, recalling that $|\Omega|=2|\Omega_+|+1$, we can define the function class of our current interest as
    \begin{align}
    \begin{split}
        \Tilde{\calH}_\Omega^{\tilde{B}} 
         \coloneqq \Bigg\{[0,2\pi)^d& \ni \xx \mapsto b_0 + \left\langle\ww,\Vec{h}(\xx)\right\rangle \,\\
        &\Bigg|\,
        \Vec{h}\in(\calM_{\Omega})^{\lvert\Omega\rvert}, \text{ and }
        b_0\in\R,\, \ww\in\R^{\lvert\Omega\rvert} \text{ such that } \lVert(b_0, \ww)\rVert_2 \leq \tilde{B}\Bigg\},
    \end{split}
    \end{align}
    where we use the notation $\langle\cdot,\cdot\rangle$ for the standard inner product. Notice how $\Tilde{\calH}_\Omega^{\tilde{B}}$ can be seen as a class of functions implemented by a single neuron with identity activation and $2$-norm bounded weights, where the input signals have been pre-processed by functions from the specified class $\calM_\Omega$.
    With this, we note the inclusion $\Fclasstrignew\subseteq\Tilde{\calH}_\Omega^{\tilde{B}}$.
    
    Next, we use Lemma~\ref{thm:2.15Wolf} (stated in Appendix~\ref{appendix:auxiliary-results}).
    To use the result, we note that the activation function of the neuron is the identity $x\mapsto x$ (which is a $1$-Lipschitz, anti-symmetric function); that $\calM_\Omega$ contains the $0$-function; that the modulus of the bias is upper bounded by $\tilde{B}$; and that we can again use Hölder's inequality applied to the $2$-norm to upper bound the $1$-norm of the weight vector as $\lVert(b_0,\ww)\rVert_1\leq \sqrt{\lvert\Omega\rvert} \lVert(b_0,\ww)\rVert_2 \leq \tilde{B}\sqrt{\lvert\Omega\rvert}.$
    With these, Lemma~\ref{thm:2.15Wolf} gives us the upper bound
    \begin{align}
        \empRad_{S\rvert_x}(\Tilde{\calH}_\Omega^{\tilde{B}}) 
        & \leq \frac{\tilde{B}}{\sqrt{m}} + 2\cdot \tilde{B}\sqrt{\lvert\Omega\rvert}\, \,\empRad_{S\rvert_x}(\calM_\Omega).\label{eq:rademacher-bound-before-massart}
    \end{align}
    Hence, in order to proceed we need to find an upper bound for the empirical Rademacher complexity of $\calM_\Omega$.
    
    We apply Massart's Lemma (which we recall as Lemma \ref{lemma:Massart} in Appendix~\ref{appendix:auxiliary-results} for completeness) for this last step.
    Let $A$ be the set of generalized trigonometric monomials with frequencies in $\Omega_+$, evaluated on every element of $S\rvert_x=(\xx_1,\ldots,\xx_m)$, i.e., 
    \begin{align}
        A 
        & \coloneqq \left\{\left(0_{\vphantom{1}},\ldots,0\right)\right\}\cup\left\{\left(\cos({\oomega}\xx_1), \ldots, \cos({\oomega}\xx_m)\right) \,|\, {\oomega}\in\Omega_+\right\} \cup \left\{
        \left(\sin({\oomega}\xx_1), \ldots, \sin({\oomega}\xx_m)\right) \,|\, {\oomega}\in\Omega_+\right\} \subseteq\R^m.
    \end{align}
    Note that, by Hölder's inequality, again applied to the $2$-norm, and since sine and cosine take values in $[-1,1]$, we have that $A\subseteq \mathcal{B}_{\sqrt{m}}(0)$, where $\mathcal{B}_r(\cc)$ is the ball of radius $r$ in $2$-norm centered at $\cc$.
    Now, we can rewrite the empirical Rademacher complexity and apply Massart's lemma (Lemma~\ref{lemma:Massart}) to get
    \begin{align}
        \empRad_{S\rvert_x}(\calM_\Omega) & \coloneqq \bbE_\sigma\left[\sup_{h\in\calM_\Omega} \frac{1}{m} \sum_{i=1}^m \sigma_i\,h(\xx_i)\right] \\
         & = \bbE_\sigma\left[\sup_{\aa\in A} \frac{1}{m}\ssigma\aa\right]\\
         & \leq \frac{\sqrt{m}}{m} \sqrt{2\log(|A|)} \\
         & \leq \frac{1}{\sqrt{m}}\sqrt{2\log(\lvert\Omega\rvert)}.
    \end{align}
    Plugging this into Eq.~\eqref{eq:rademacher-bound-before-massart}, we obtain
    \begin{align}
        \empRad_{S\rvert_x}(\Tilde{\calH}_\Omega^{\tilde{B}}) 
        & \leq \frac{\tilde{B}}{\sqrt{m}} + 2\cdot \tilde{B}\sqrt{\lvert\Omega\rvert} \cdot\frac{1}{\sqrt{m}} \sqrt{2\log(\lvert\Omega\rvert)} \\
        & \leq \mathcal{O}\left(\frac{\tilde{B}}{\sqrt{m}}\sqrt{\lvert\Omega\rvert\log(\lvert\Omega\rvert)}\right).
    \end{align}
    Recalling again that $\Fclasstrignew\subseteq \Tilde{\calH}_\Omega^{\tilde{B}}$ then yields the claimed bound.
\end{proof}

\begin{proof}[Proof of Lemma~\ref{lemma:rademacher-complexity-trig-polys}]
    This follows directly from combining Lemmas~\ref{lemma:Old_V1} and~\ref{lemma:Old_V2}.
\end{proof}

With this Rademacher complexity bound at hand, we can make use of standard tools from classical statistical learning theory to derive a generalization bound.

\begin{theorem}[Generalization bound for GTPs---Version $1$]\label{theorem:rademacher-gen-bound-trig-polys}
Let $d,m\in\mathbb{N}$. Let $\Fclasstrignew$ be as defined in Eq.~\eqref{eq:trig_polys_class}. Let $\ell:\R\times\R\to [0,c]$ be a bounded loss function such that~$\mathbb{R}\ni z\mapsto\ell(y,z)$ is $L$-Lipschitz for all $y\in\mathbb{R}$. 
For any $\delta\in (0,1)$ and for any probability measure $P$ on $[0,2\pi)^d\times \R$, with probability $\geq 1-\delta$ over the choice of i.i.d.~training data $S=\{(\xx_i,y_i)\}_{i=1}^m\in ([0,2\pi)^d\times\R)^m$ of size $m$, for every $f\in\Fclasstrignew$, the generalization error can be upper-bounded as
\begin{align}
    R(f) - \hat{R}_S(f)
    &\leq \mathcal{O}\left( \frac{L\min\left\{\max\{K,\tilde{B}\sqrt{\lvert\Omega\rvert}\}\sqrt{\log(2d)}, \tilde{B}\sqrt{\lvert\Omega\rvert\log(\lvert\Omega\rvert)} \right\}}{\sqrt{m}} + \frac{\sqrt{\log(\nicefrac{1}{\delta})}}{\sqrt{m}}\right).
\end{align}
\end{theorem}
\begin{proof}
The proof of this theorem consists in combining the standard generalization bound in terms of Rademacher complexity with the Rademacher complexity bounds from Lemma \ref{lemma:rademacher-complexity-trig-polys}. More precisely, we define $\smash{\mathcal{G}\subseteq [0,c]^{[0,2\pi)^d\times\mathbb{R}}}$ to be the class of functions that can be obtained by post-composing elements of $\Fclasstrignew$ with the loss function $\ell$ -- i.e. we define
\begin{align}
    \mathcal{G}
    &\coloneqq \left\{ [0,2\pi)^d\times\mathbb{R}\ni (\xx,y)\mapsto \ell(y,f(\xx))~|~f\in\Fclasstrignew \right\}.
\end{align}
We then have the following generalization bound (see, e.g., Theorem $3.3$ in
Ref.~\cite{MohriRostamizadehTalwalkar18} or Theorem $1.15$ in
Ref.~\cite{Wolf.2020}): For any probability measure $P$ on $[0,2\pi)^d\times \mathbb{R}$ and for any $\delta>0$, with probability $\geq 1-\delta$ over the choice of an i.i.d.~training data set $S=\{(\xx_i,y_i)\}_{i=1}^m\in([0,2\pi)^d\times\R)^m$ of size $m$ drawn according to $P$, we have, for every $g\in\mathcal{G}$,
\begin{align}
    \mathbb{E}_{(\xx,y)\sim P} [g(\xx,y)] - \frac{1}{m}\sum\limits_{i=1}^m g(\xx_i,y_i)
    &\leq 2\hat{\mathcal{R}}_S (\mathcal{G}) + 3c\sqrt{\frac{\log(\nicefrac{2}{\delta})}{2m}}.\label{eq:standard-rademacher-gen-bound}
\end{align}
Note that, when writing $g\in\mathcal{G}$ as $g(\xx,y)=\ell(y,f(\xx))$ for some $f\in\Fclasstrignew$, we directly have
\begin{align}
    \mathbb{E}_{(\xx,y)\sim P} [g(\xx,y)] - \frac{1}{m}\sum\limits_{i=1}^m g(\xx_i,y_i)
    &= R(f)-\hat{R}_S(f).
\end{align}
That is, Eq.~\eqref{eq:standard-rademacher-gen-bound} indeed provides a high-probability bound on the generalization error. Therefore, we now upper-bound the empirical Rademacher complexity $\hat{\mathcal{R}}_S (\mathcal{G})$. To this end, we use Talagrand's Lemma (going back to Ref.~\cite{ledoux1991probability}) and our bounds for the empirical Rademacher complexity of $\Fclasstrignew$.
As we assume that $\mathbb{R}\ni z\mapsto\ell(y,z)$ is $L$-Lipschitz for all $y\in\mathbb{R}$, we can apply Talagrand's Lemma (Lemma \ref{lemma:talagrand}) and Lemma \ref{lemma:rademacher-complexity-trig-polys} to obtain 
\begin{align}
    \hat{\mathcal{R}}_S (\mathcal{G})
    &= \frac{1}{m}\mathbb{E}_\sigma\left[\sup\limits_{g\in\mathcal{G}} \sum\limits_{i=1}^m \sigma_i g(x_i,y_i)\right]\\
    &= \frac{1}{m}\mathbb{E}_\sigma\left[\sup\limits_{f\in\Fclasstrignew} \sum\limits_{i=1}^m \sigma_i \ell(y_i,f(x_i))\right]\\
    &\leq \frac{L}{m}\mathbb{E}_\sigma\left[\sup\limits_{f\in\Fclasstrignew} \sum\limits_{i=1}^m \sigma_i f(x_i)\right]\\
    &= L\hat{\mathcal{R}}_{S\vert_x} (\Fclasstrignew)\\
    &\leq \mathcal{O}\left(L\frac{\min\left\{\sqrt{\log(2d)}\max\{K,\tilde{B}\sqrt{\lvert\Omega\rvert}\}, \tilde{B}\sqrt{\lvert\Omega\rvert\log(\lvert\Omega\rvert)}\right\}}{\sqrt{m}} \right),
\end{align}
where we have denoted by $S\vert_x\coloneqq\{\xx_i\}_{i=1}^m$ the set of unlabeled training data points.
Inserting this bound into Eq.~\eqref{eq:standard-rademacher-gen-bound} now gives the stated generalization error bound.
\end{proof}

The generalization bound of Theorem \ref{theorem:rademacher-gen-bound-trig-polys} can be rewritten as an upper bound on the number of labeled training examples that suffice to guarantee small generalization error.

\begin{corollary}[Number of labeled training 
examples sufficient for a small generalization error---Version $1$]
 \label{corollary:rademacher-sample-complexity-bound-trig-polys}
For any $\varepsilon,\delta\in (0,1)$ and for any probability measure $P$ on $[0,2\pi)^d\times \mathbb{R}$, a training data size
\begin{align}
    m = m(\varepsilon,\delta)
    &\leq \mathcal{O}\left( \frac{L^2\min\left\{\max\{K^2,\tilde{B}^2\lvert\Omega\rvert\}\log(2d), \tilde{B}^2\lvert\Omega\rvert\log(\lvert\Omega\rvert) \right\}}{\varepsilon^2} + \frac{c^2\log(\nicefrac{1}{\delta})}{\varepsilon^2}\right)
\end{align}
suffices to guarantee that, with probability $\geq 1-\delta$ over the choice of i.i.d.~training data $S\in([0,2\pi)^d\times \mathbb{R})^m$ of size $m$, $R(f)- \hat{R}_S(f)\leq\varepsilon$ holds for every $f\in\Fclasstrignew$. 
\end{corollary}
\begin{proof}
We set the upper bound on the generalization error proven in Theorem \ref{theorem:rademacher-gen-bound-trig-polys} equal to $\varepsilon$ and solving for $m$.
\end{proof}

\begin{remark}\label{rmk:rademacher-p-norms}
The proof strategy for obtaining Rademacher complexity bounds of generalized trigonometric polynomials presented here easily extends beyond the case in which the $2$-norm of the vector of Fourier coefficients is assumed to be bounded. Namely, if we consider, for $1\leq p\leq\infty$, the class
\begin{align}
    \mathcal{H}_\Omega^{\tilde{B},p}
    \coloneqq \left\{[0,2\pi)^d\ni \xx\mapsto \frac{a_0}{2} + \sum\limits_{\oomega\in\Omega_+} \left(a_{\oomega}\cos(\oomega\xx) + b_{\oomega}\sin(\oomega\xx)\right)~\Bigg\vert~ \norm{(a_0,(a_{\oomega})_{\oomega},(b_{\oomega})_{\oomega})}_p \leq \tilde{B}\right\},
\end{align}
with Fourier coefficients of a bounded $p$-norm, we obtain, with essentially the same proof, an empirical Rademacher complexity bound of
\begin{align}
    \hat{\mathcal{R}}_{S\rvert_x} (\mathcal{H}_\Omega^{\tilde{B},p})
    &\leq \tilde{\mathcal{O}}\left(\frac{\tilde{B}\lvert\Omega\rvert^{\frac{1}{q}}}{\sqrt{m}} \right),
\end{align}
where $q\in [0,1]$ is the Hölder conjugate of $p$, i.e., $\nicefrac{1}{p} + \nicefrac{1}{q} = 1$, and the $\tilde{\mathcal{O}}$ hides a logarithmic dependence on $\lvert\Omega\rvert$. This, in turn, leads (for $c$-bounded $L$-Lipschitz loss) to a generalization error bound of
\begin{align}
    R(f) - \hat{R}_S(f)
    &\leq \tilde{\mathcal{O}}\left( \frac{L\tilde{B}\lvert\Omega\rvert^{\frac{1}{q}} + c\sqrt{\log(\nicefrac{1}{\delta})}}{\sqrt{m}}\right),
\end{align}
which holds with probability $\geq 1-\delta$ uniformly over $\mathcal{H}_\Omega^{\tilde{B},p}$, for training data of size $m$.
These bounds based on $p$-norms might be of independent interest. For example, depending on the structure of the trainable part of the PQC, a detailed analysis might lead to additional structural properties (such as sparsity) of the set of admissible Fourier coefficients, which could then lend themselves to an analysis in terms of $p$-norms for $p\neq 2$.
\end{remark}

\subsection{Generalization bounds for generalized trigonometric polynomials via covering numbers}\label{sbs:gen-bounds-trig-poly-covering}

Similarly to
Section~\ref{sbs:gen-bounds-trig-poly-rademacher}, we first prove a bound on a complexity measure for the 
hypothesis 
class 
$\Fclass$ and then derive a generalization bound from it.
This subsection differs from the previous one in that we discuss a different complexity measure, covering numbers.

\begin{lemma}[Covering number bound for GTPs]\label{lemma:covering-trig-polys}
Let $d\in\mathbb{N}$ and $\varepsilon>0$. Let 
$\Fclass$
be as defined in Eq.~
\eqref{eq:function_class_bounded_inf_norm}.
Also, let $\tilde{B}>0$ be such that $\Fclass\subseteq\Fclasstrignew$.
The $\varepsilon$-covering number of 
$\Fclass$
with respect to~$\norm{\cdot}_\infty$ can be upper-bounded as
\begin{align}
    \mathcal{N}(\Fclass,\norm{\cdot}_\infty,\varepsilon)
    &\leq \mathcal{N}(\Fclasstrignew,\norm{\cdot}_\infty,
    \nicefrac{\varepsilon}{2}
    )
    \leq \left(\frac{
    2
    \cdot 3\cdot \tilde{B}\sqrt{\lvert\Omega\rvert}}{\varepsilon}\right)^{\lvert\Omega\rvert}.
\end{align}
Therefore, the corresponding metric entropy can be upper-bounded as
\begin{align}
    \log_2 \mathcal{N}(
    \Fclass
    ,\norm{\cdot}_\infty,\varepsilon)
    &\leq \mathcal{O}\left(\lvert\Omega\rvert[\log(\tilde{B}) + \log(\lvert\Omega\rvert) + \log(\nicefrac{1}{\varepsilon})] \right).
\end{align}
\end{lemma}
\begin{proof}
    By our assumption on the choice of $\tilde{B}$, we have $\Fclass\subseteq\Fclasstrignew$. Therefore, according to the approximate monotonicity of covering numbers (see, e.g., Exercise $4.2.10$ in~\cite{Vershynin.2018}), we have, for every $\varepsilon >0$,
    \begin{align}
        \mathcal{N}(\Fclass,\norm{\cdot}_\infty,\varepsilon)
        &\leq \mathcal{N}(\Fclasstrignew,\norm{\cdot}_\infty,\nicefrac{\varepsilon}{2}).
    \end{align}
    Thus, it remains to prove a covering number bound for $\Fclasstrignew$.
    
    Let $\calN_{\teps}$ be an~$\tilde{\varepsilon}$-covering net of the ball
    \begin{align}
        \calB & \coloneqq \left\{\xxi=(a_0,(a_{\oomega})_{\oomega\in\Omega_+},(b_{\oomega})_{\oomega\in\Omega_+})\in\R^{\lvert\Omega\rvert} \,\Big|\, \norm{\xxi}_2\leq \tilde{B}\right\}
    \end{align}
    with respect to the metric induced by $\norm{\cdot}_2$ on $\R^{\lvert\Omega\rvert}$. By definition of $\Fclasstrignew$, to every $f\in\Fclasstrignew$ we can associate a point $(a_0,(a_{\oomega})_{\oomega\in\Omega_+},(b_{\oomega})_{\oomega\in\Omega_+})\in\calB$ such that
    \begin{align}
        f(x) = \frac{a_0}{2} + \sum_{\oomega\in\Omega_+} 
        \left(
        a_{\oomega}
        \cos(\oomega\xx) + b_{\oomega} \sin(\oomega\xx)\right).
    \end{align}
    Given such a vector of coefficients $(a_0,(a_{\oomega})_{\oomega\in\Omega_+},(b_{\oomega})_{\oomega\in\Omega_+})\in\calB$ -- which, again, for the sake of notational ease, we write as $(a_0,(a_{\oomega})_{\oomega},(b_{\oomega})_{\oomega}))$, omitting the $\Omega_+$ everywhere -- we can find an element $(\Tilde{a}_0, (\Tilde{a}_{\oomega})_{\oomega\in\Omega_+}, (\Tilde{b}_{\oomega})_{\oomega\in\Omega_+})\in\calN_{\teps}$ of the cover that is $\teps$ close in $2$-norm to the coefficients of $f$, i.e., such that
    \begin{align}
        \norm{(a_0,(a_{\oomega})_{\oomega},(b_{\oomega})_{\oomega}) - (\Tilde{a}_0, (\Tilde{a}_{\oomega})_{\oomega}, (\Tilde{b}_{\oomega})_{\oomega})}_2 \leq\teps.
    \end{align}
    Define $\Tilde{f}$ as the function specified by these new coefficients,
    \begin{align}
        \Tilde{f}(x) 
        & = \frac{\tilde{a}_0}{2} + \sum_{\oomega\in\Omega_+}
        \left(
        \Tilde{a}_{\oomega} \cos(\oomega\xx) + \Tilde{b}_{\oomega} \sin(\oomega\xx)
        \right).
    \end{align}
    We now bound the infinity norm distance between $f$ and $\Tilde{f}$ in terms of the $2$-norm distance between the corresponding coefficients as
    \begin{align}
        \norm{f-\Tilde{f}}_\infty 
        & \coloneqq \sup_{\xx\in [0,2\pi)} \left\lvert f(\xx) - \Tilde{f}(\xx) \right\rvert \\
        & \leq \left\lvert \frac{a_0}{2}- \frac{\tilde{a}_0}{2} \right\rvert + \sup_{\xx} \sum_{\oomega\in\Omega_+} \left\lvert (a_{\oomega} - \Tilde{a}_{\oomega})\cos(\oomega\xx) + (b_{\oomega} - \Tilde{b}_{\oomega})\sin(\oomega\xx) \right\rvert \\
        & \leq \lvert a_0 - \Tilde{a}_0 \rvert + \sum_{\oomega\in\Omega_+}
        \left(\lvert a_{\oomega} - \Tilde{a}_{\oomega}\rvert + \lvert b_{\oomega} - \Tilde{b}_{\oomega}
        \rvert\right) \\
        & = \norm{(a_0,(a_{\oomega})_{\oomega},(b_{\oomega})_{\oomega}) - (\Tilde{a}_0,(\Tilde{a}_{\oomega})_{\oomega},(\Tilde{b}_{\oomega})_{\oomega})}_1 \\
        & \leq \sqrt{\lvert\Omega\rvert} \, \teps.
    \end{align}
    Here, we have used the triangle inequality and the fact that sine and cosine can only take values in $[-1,1]$, as well as (in the last step) Hölder's inequality with respect to the $2$-norm.
    That means, if we denote by $\calN_{\calF}$ the set of GTPs whose coefficients come from the cover $\calN_{\teps}$, i.e.,
    \begin{align}
        \calN_{\calF} 
         \coloneqq \Bigg\{[0,2\pi)^d \ni\xx\mapsto \frac{\tilde{a}_0}{2} + \sum_{\oomega\in\Omega_+} \left(a_{\oomega}\cos(\oomega\xx) + b_{\oomega}\sin(\oomega\xx)\right) 
         \Bigg|~
        (\Tilde{a}_0,(\Tilde{a}_{\oomega})_{\oomega},(\Tilde{b}_{\oomega})_{\oomega}) \in \calN_{\teps}\Bigg\},
    \end{align}
    and if we fix $\teps$ to be $\teps=\varepsilon/\sqrt{\lvert\Omega\rvert}$, then $\calN_{\calF}$ is an $\varepsilon$-covering net of $\Fclasstrignew$ with respect to~$\norm{\cdot}_\infty$.\\
    Thus, to finish the proof, it remains to upper bound the cardinality $\lvert \calN_{\calF}\rvert \leq \lvert \calN_{\teps}\rvert$.
    To obtain such a bound, we recall that we only require $\calN_{\teps}$ to be an $\teps$-cover of a $2$-norm ball of radius $\tilde{B}$ in $\mathbb{R}^{\lvert\Omega\rvert}$ with respect to~the $2$-norm. A simple volumetric argument (presented, e.g., in section $4$ of Ref.~\cite{Vershynin.2018}) shows that there exists such a $\teps$-cover $\calN_{\teps}$ of $\calB$ with cardinality 
    \begin{align}
        \lvert \calN_{\teps}\rvert  
        & \leq \left(\frac{3\cdot \tilde{B}}{\teps}\right)^{\lvert\Omega\rvert}
        = \left(\frac{3 \cdot \tilde{B}\sqrt{\lvert\Omega\rvert}}{\varepsilon}\right)^{\lvert\Omega\rvert}.
    \end{align}
    All in all, we have proven that there exists an $\varepsilon$-covering net of $\Fclasstrignew$ with respect to~$\norm{\cdot}_\infty$ whose cardinality is bounded by 
    \begin{align}
        \left(\frac{3\cdot \tilde{B}\sqrt{\lvert\Omega\rvert}}{\varepsilon}\right)^{\lvert\Omega\rvert}.
    \end{align}
    This is exactly the claimed upper bound on the $\varepsilon$-covering number of $\Fclasstrignew$, thus completing the proof.
\end{proof}
The covering number bound just established implies a generalization bound for GTPs.

\begin{theorem}[Generalization bound for generalized trigonometric polynomials---Version $2$]\label{theorem:covering-gen-bound-trig-polys}
Let $d,m\in\mathbb{N}$. Let 
$\Fclass$
be as defined in Eq.~
\eqref{eq:function_class_bounded_inf_norm}.
Let $\tilde{B}>0$ be such that $\Fclass\subseteq\Fclasstrignew$.
Let $\ell:\R\times\R\to [0,c]$ be a bounded loss function such that~$\mathbb{R}\ni z\mapsto\ell(y,z)$ is $L$-Lipschitz for all $y\in\mathbb{R}$.
For any $\delta\in (0,1)$ and for any probability measure $P$ on $[0,2\pi)^d\times \R$, with probability $\geq 1-\delta$ over the choice of i.i.d.~training data $S\in ([0,2\pi)^d\times\R)^m$ of size $m$, for every $f\in
\Fclass
$, the generalization error can be upper-bounded as
\begin{align}
    R(f) - \hat{R}_S(f)
    &\leq \mathcal{O}\left( 
    BL\sqrt{\frac{\lvert\Omega\rvert (\log(\lvert\Omega\rvert) + \log( \tilde{B}))}{m}} + c\, \sqrt{\frac{\log(\nicefrac{1}{\delta})}{m}}\right)
\end{align}
\end{theorem}
\begin{proof}
The proof consists of three steps. First, we use the chaining technique from random process theory to upper bound the (empirical) Rademacher complexity in terms of an integral over the square root of the uniform empirical metric entropy. Second, we show that the metric entropy with respect to~$\norm{\cdot}_\infty$ upper-bounds the uniform empirical metric entropy, so we can use the bound in Lemma \ref{lemma:covering-trig-polys} to upper-bound the (empirical) Rademacher complexity of generalized trigonometric polynomials. Third, we again use the standard generalization bound based on empirical Rademacher complexities.\\

Similarly to
the proof of Theorem \ref{theorem:rademacher-gen-bound-trig-polys}, we define
\begin{align}
    \mathcal{G}
    &\coloneqq \left\{[0,2\pi)^d\times\mathbb{R}\ni (\xx,y)\mapsto \ell(y,f(\xx))~\big|~f\in
    \Fclass
    \right\}.
\end{align}
Again, since we assume that $\mathbb{R}\ni z\mapsto\ell(y,z)$ is $L$-Lipschitz for all $y\in\mathbb{R}$, Talagrand's Lemma (Lemma~\ref{lemma:talagrand} in Appendix~\ref{appendix:auxiliary-results}) tells us that
\begin{align}
    \hat{\mathcal{R}}_S (\mathcal{G})
    \leq L\hat{\mathcal{R}}_{S\vert_x} (\Fclass)
    ,
\end{align}
where we have denoted by $S\vert_x\coloneqq\{\xx_i\}_{i=1}^m$ the 
unlabeled training data points.
Next, Dudley's Theorem (which we recall as Theorem \ref{theorem:dudley} 
in Appendix~\ref{appendix:auxiliary-results}), yields
\begin{align}
    \hat{\mathcal{R}}_{S\vert_x} (
    \Fclass
    )
    &\leq \frac{12}{\sqrt{m}}\int\limits_{0}^{\gamma_0}\sqrt{\log\mathcal{N}(
    \Fclass,\norm{\cdot}_{2,S\vert_x},\beta)}
    ~\mathrm{d}\beta,\label{eq:dudley-step-1}
\end{align}
where $\norm{\cdot}_{2,S\vert_x}$ is the (data-dependent) semi-norm on $\mathbb{R}^{\mathbb{R}^d}$ defined as $\norm{f}_{2,S\vert_x} := \left(\frac{1}{m}\sum_{i=1}^m \lvert f(\xx_i)\rvert^2\right)^{\nicefrac{1}{2}}$, and we have defined $\gamma_0\coloneqq\sup_{f\in
\Fclass}
\norm{f}_{2,S}$.

Now, we note that, for every $f\in
\Fclass
$, $\norm{f}_{2,S\rvert_x}\leq \norm{f}_\infty$. Therefore, we have both that $\gamma_0\leq \sup_{f\in
\Fclass}
\norm{f}_\infty\leq 
B$ and, that for every $\beta>0$, $\mathcal{N}(
\Fclass
,\norm{\cdot}_{2,S\rvert_x},\beta)\leq \mathcal{N}(
\Fclass
,\norm{\cdot}_{\infty},\beta)$. Hence, we can combine Eq.~\eqref{eq:dudley-step-1} with our covering number bound from Lemma \ref{lemma:covering-trig-polys} and further upper bound
\begin{align}
    \hat{\mathcal{R}}_{S\vert_x} (
    \Fclass
    )
    &\leq \frac{12}{\sqrt{m}}\int\limits_{0}^{\gamma_0} \sqrt{\lvert\Omega\rvert\left(\log(3\cdot \tilde{B}) + \log(\sqrt{\lvert\Omega\rvert}) + \log\left(\frac{
    2
    }{\beta}\right) \right)}~\mathrm{d}\beta\\
    &\leq \frac{12}{\sqrt{m}} \sqrt{\lvert\Omega\rvert}\left( \gamma_0\sqrt{\log(3\cdot \tilde{B}) + \frac{1}{2} \log(\lvert\Omega\rvert)} +  \int\limits_{0}^{\gamma_0} \sqrt{\log\left(\frac{
    2
    }{\beta}\right)}~\mathrm{d}\beta\right)\\
    &\leq \frac{12}{\sqrt{m}} \sqrt{\lvert\Omega\rvert}\Bigg( 
    B\sqrt{\log(3\cdot \tilde{B}) + \frac{1}{2}\log(\lvert\Omega\rvert)} \\
    &\hphantom{\leq\frac{12}{\sqrt{m}} \sqrt{\lvert\Omega\rvert}\Bigg( ~}+ 
    B\sqrt{\log\left(\frac{1}{\tilde{B}}\right)} - \frac{\sqrt{\pi}}{2}\operatorname{erf}\left(\sqrt{\log\left(\frac{1}{\tilde{B}}\right)}\right) 
    \Bigg)\\
    &\leq \mathcal{O}\left( 
    B\sqrt{\frac{\lvert\Omega\rvert (\log(\tilde{B}) + \log(\lvert\Omega\rvert))}{m}}\right),
\end{align}
where we have used the integral \begin{equation}
\int\sqrt{\log\nicefrac{1}{x}}~\mathrm{d}x = x\sqrt{\log\nicefrac{1}{x}} - (\nicefrac{\sqrt{\pi}}{2}) \cdot \operatorname{erf}(\sqrt{\log\nicefrac{1}{x}}), 
\end{equation}
with the error function defined as 
\begin{equation}
\operatorname{erf}(x):=\frac{2}{\sqrt{\pi}}\int_{0}^{x}\exp(-t^2)~\mathrm{d}t.
\end{equation}
At this point, we again have a bound on the empirical Rademacher complexity at our disposal. So, just like in the proof of Theorem \ref{theorem:rademacher-gen-bound-trig-polys}, we can now apply the standard Rademacher complexity generalization bound. This then tells us that, for any probability measure $P$ on $[0,2\pi)^d\times\mathbb{R}$ and for any $\delta>0$, with probability $\geq 1-\delta$ over the choice of an i.i.d.~training data set $S$ of size $m$, we have, for every $f\in
\Fclass
$,
\begin{align}
    R(f)- \hat{R}_S(f)
    &\leq 2\hat{\mathcal{R}}_S(\mathcal{G}) + 3c\sqrt{\frac{\log(\nicefrac{2}{\delta})}{2m}}\\
    &\leq \mathcal{O}\left( 
    BL\sqrt{\frac{\lvert\Omega\rvert (\log(\lvert\Omega\rvert) + \log( \tilde{B}))}{m}} + c\sqrt{\frac{\log(\nicefrac{1}{\delta})}{m}}\right),
\end{align}
as claimed.
\end{proof}

Also for this generalization bound, we provide the reformulation in terms of a bound on the sample size sufficient to guarantee small generalization error.

\begin{corollary}[Number of labeled training examples sufficient for a small generalization error---Version $2$]\label{corollary:covering-sample-complexity-bound-trig-polys}
Let $d\in\mathbb{N}$. Let $
\Fclass$ be as defined in Eq.~\eqref{eq:function_class_bounded_inf_norm}. Let $\tilde{B}>0$ be such that $\Fclass\subseteq\Fclasstrignew$.
Let $\ell:\R\times\R\to [0,c]$ be an $L$-Lipschitz loss function. 
For any $\varepsilon,\delta\in (0,1)$ and for any probability measure $P$ on $[0,2\pi)^d\times \mathbb{R}$, a training data size
\begin{align}
    m = m(\varepsilon,\delta)
    &\leq \mathcal{O}\left( 
    B^2 L^2 \frac{\lvert\Omega\rvert (\log(\lvert\Omega\rvert) + \log(\tilde{B}))}{\varepsilon^2} + c^2\frac{\log(\nicefrac{1}{\delta})}{\varepsilon^2}\right),
\end{align}
suffices to guarantee that, with probability $\geq 1-\delta$ over the choice of i.i.d.~training data $S\in(\mathbb{R}^d\times \mathbb{R})^m$ of size $m$, $R(f)- \hat{R}_S(f)\leq\varepsilon$ holds for every $f\in
\Fclass$.
\end{corollary}
\begin{proof}
We set the upper bound on the generalization error proven in Theorem \ref{theorem:covering-gen-bound-trig-polys} equal to $\varepsilon$ and solving for $m$.
\end{proof}

\begin{remark}\label{rmk:covering-p-norms}
Our metric entropy bounds of trigonometric polynomials presented here again extend beyond the case of bounded $2$-norm of the vector of Fourier coefficients to a general bounded $p$-norm. However, if we again consider, for $1\leq p\leq\infty$, the class $\mathcal{H}_\Omega^{\tilde{B},p}$ defined in Remark \ref{rmk:rademacher-p-norms}, our proof strategy here yields essentially -- i.e., to leading order in $\lvert\Omega\rvert$ -- the same metric entropy and generalization bounds as for $p=2$. The reason is that the dimension of the space in which we take covering nets in the proof of Lemma \ref{lemma:covering-trig-polys} remains proportional to $\lvert\Omega\rvert$, independently of $p$. We only see improvements for $1\leq p<2$ in the terms depending logarithmically on $\lvert\Omega\rvert$.
Therefore, in the case of a constant $\tilde{B}$, while the proof strategies of Sections \ref{sbs:gen-bounds-trig-poly-rademacher} and \ref{sbs:gen-bounds-trig-poly-covering} give essentially the same generalization guarantees for $p=2$, the approach of Section~\ref{sbs:gen-bounds-trig-poly-rademacher} adapts nicely to the case $p<2$, whereas the reasoning of Section~\ref{sbs:gen-bounds-trig-poly-covering} is typically preferable for $p>2$.
\end{remark}

\begin{remark}\label{rmk:hoelder-continuous-loss}
The proof of Theorem \ref{theorem:covering-gen-bound-trig-polys} extends beyond Lipschitz loss functions. For example, suppose that $\mathbb{R}\ni z\mapsto\ell(y,z)$ is $\alpha$-Hölder continuous with Hölder coefficient $A>0$ for all $y\in\mathbb{R}$, where $\alpha\in (0,1)$. Then, with the notation of the above proof,
\begin{align}
    \mathcal{N}(\mathcal{G}, \norm{\cdot}_{2,S},\beta)
    \leq \mathcal{N}\left(
    \Fclass
    , \norm{\cdot}_{2\alpha,S\rvert_x},\left(\nicefrac{\beta}{A}\right)^{\nicefrac{1}{\alpha}}\right).
\end{align}
We can thus apply Dudley's Theorem to upper bound
\begin{align}
    \hat{\mathcal{R}}_S(\mathcal{G})
    &\leq \frac{12}{\sqrt{m}}\int\limits_{0}^{\tilde{\gamma}_0} \sqrt{\mathcal{N}\left(
    \Fclass
    , \norm{\cdot}_{2\alpha,S\rvert_x},\left(\nicefrac{\beta}{A}\right)^{\nicefrac{1}{\alpha}}\right)}~\mathrm{d}\beta.
\end{align}
Now, we again observe that $\norm{\cdot}_{2\alpha,S\rvert_x}\leq \norm{\cdot}_\infty$ and upper bound the covering number integral, using our result from Lemma \ref{lemma:covering-trig-polys}. The parameters of the Hölder continuity enter the final Rademacher complexity bound via a term scaling with $\sqrt{\nicefrac{\log(A)}{\alpha}}$.
\end{remark}

\section{Encoding-dependent generalization bounds for parametrized quantum circuits}\label{s:gen-bounds-pqc}

We are finally in a position to answer the questions posed in Section~\ref{s:motivation}. Recall that our first goal was to derive generalization bounds for PQC-based models which depend explicitly on architectural hyper-parameters related to the data-encoding strategy. We showed in Section~\ref{s:pqc_models} how PQC-based model classes can be viewed as a subset of generalized trigonometric polynomials (GTPs), whose set of frequencies $\Omega$ is determined solely by the data-encoding strategy $\calD$. We then derived complexity and generalization bounds for GTPs in terms of the number of different frequencies $|\Omega(\calD)|$. In order to provide explicitly encoding-dependent generalization bounds for PQC-based models, it remains to express $|\Omega(\calD)|$ in terms of the relevant architectural hyper-parameters associated with different data-encoding strategies.

To do so, we recall that the data-encoding strategy of a PQC-based model class is defined as a collection of lists of data-encoding Hamiltonians \smash{$\calD^{(i)} = \{H_j^{(i)}\}$} associated with each coordinate $x^{(i)}$. We distinguish different data-encoding strategies according to the different assumptions made on the structure of the data-encoding Hamiltonians $H \in \calD^{(i)}$. Given a particular assumption, for example that all $H$ are tensor products of Pauli operators or at most $\kappa$-local, the natural hyper-parameter associated with the data encoding strategy is the number \smash{$N=\sum_{i=1}^d |\calD^{(i)}|$} of data-encoding Hamiltonians of the assumed type. Hence, our goal in this section is to derive, for different data-encoding strategies, upper bounds on $|\Omega(\calD)|$ that depend on $N$ as well as as on other relevant properties of the data-encoding Hamiltonians (such as, e.g., the locality $\kappa$). By substituting these upper bounds on $|\Omega|$ into the GTP generalization bounds of the previous section, we then obtain generalization bounds for PQC-based model classes which depend explicitly on properties of the data-encoding strategy.

We first recall the definition of $\Omega$ from Eq.~\eqref{eqn:omega_full_definition}. If we denote the Hamiltonians of the data-encoding strategy associated with $x^{(i)}$ as \smash{$\{ H_j^{(i)} \}$}, we can group the frequencies associated with each data coordinate into a separate sumset $\Omega^{(i)}$:
\begin{align}
    \Omega(\calD) = \sum_{i=1}^{d} \sum_{j=1}^{N^{(i)}} \Omega\left( H_{j}^{(i)} \right) = \sum_{i=1}^{d} \Omega^{(i)}.
\end{align}
The frequencies belonging to the different coordinates $\{x^{(i)}\}$ are linearly independent because they were defined to be multiples of different standard basis vectors $\ee^{(i)}$. This implies that the cardinality of the full set is equal to the product of the individual cardinalities,
\begin{align}
    |\Omega| = \prod_{i=1}^d |\Omega^{(i)}|,
\end{align}
thus allowing us to multiply bounds on the cardinalities obtained for the separate data-encoding strategies, $|\Omega^{(i)}|$, to obtain a bound on $|\Omega|$.

As the underlying frequencies in $\Omega^{(i)}$ are all scalar multiples of the same basis vector $\ee^{(i)}$, the analysis of $\Omega^{(i)}$ comes down to the different frequencies generated by the Hamiltonians that are used to encode $x^{(i)}$. For a given single Hamiltonian $H$, we denote this set by
\begin{align}
    \Delta(H) \coloneqq \{ \lambda_i - \lambda_j \, | \, \lambda_i, \lambda_j \in \operatorname{spec}(H) \}
\end{align}
so that
\begin{align}
    \Omega(H) = \left\{\left.\delta \ee \, \right| \, \delta \in \Delta\left(H\right) \right\},
\end{align}
where $\ee$ is the basis vector associated to the respective coordinate. Next, we derive some bounds on $|\Omega^{(i)}|$ for different assumptions on the underlying Hamiltonians.

\textbf{Worst case upper bounds.}
We first derive the worst-case limits of $|\Omega^{(i)}|$ for $\kappa$-local encoding Hamiltonians. A $\kappa$-local Hamiltonian $H$ has local dimension $2^{\kappa}$ 
and the number of possible differences of eigenvalues in the spectrum is thus upper bounded as
\begin{align}\label{eqn:upper_bound_l_local_spectrum}
    |\Delta(H)| \leq \frac{2^{\kappa}(2^{\kappa} - 1)}{2}+1 = \calO(2^{2{\kappa}}).
\end{align}
One can in principle construct a Hamiltonian that saturates this bound by choosing $\operatorname{spec}(H_{\max}) = \{0, 3, 9, \dots, 3^{2^{\kappa}}\}$, but this is a rather synthetic example that we do not expect to encounter on real hardware.
Eq.~\eqref{eqn:upper_bound_l_local_spectrum} implies that repeating $N^{(i)}$ $\kappa$-local Hamiltonians will, in the case where there are no duplicates in the frequency set, imply a cardinality of at most
\begin{align}
    |\Omega^{(i)}| 
    \leq \left(\frac{2^{\kappa}(2^{\kappa} - 1)}{2}+1\right)^{N^{(i)}} = \calO(2^{2{\kappa} N^{(i)}}).
\end{align}
Again, this bound can be saturated by choosing Hamiltonians with ever-larger spectra, namely by choosing $H_{1}^{(i)} = H_{\max}$ and $\operatorname{spec}(H_{j+1}^{(i)}) = \max(\operatorname{spec}(H_{j}^{(i)})) \cdot \operatorname{spec}(H_{\max})$.

\textbf{Repeated Hamiltonians.}
We now consider the case where the same Hamiltonian $H^{(i)}$ is used $N^{(i)}$ times to encode the coordinate $x^{(i)}$.
Due to the underlying symmetry of the definition of $\Delta(H^{(i)})$, we have that
\begin{align}
    \Delta(H^{(i)}) = \{ 0, \pm \delta_1, \dots, \pm \delta_T \}
\end{align}
for some $T$, and therefore $|\Delta(H^{(i)})| = 2T+1 = |\Omega_j^{(i)}| = |\Omega_0^{(i)}|$, where we have denoted the repeated set of frequencies common to all encoding gates as $\Omega_0^{(i)}$. Using the results on the maximum size of the spectrum of a $\kappa$-local Hamiltonian in Eq.~\eqref{eqn:upper_bound_l_local_spectrum}, we can deduce that $T \leq 2^{\kappa-2}(2^{\kappa}-1)$.
We now quantify the number of different frequencies in $\Omega^{(i)}$ in terms of $T$. 
$N^{(i)}$ repetitions of the fixed Hamiltonian with frequencies $\Omega_0^{(i)}$ result in a set of frequencies that contains all possible combinations of $N^{(i)}$ vectors $\oomega_j$ from $\Omega_0^{(i)}$:
\begin{align}
    \Omega^{(i)} &= \left\{ \left.\sum_{j=1}^{N^{(i)}} \oomega_j \, \right| \, \oomega_j \in \Omega_0^{(i)} \text{ for all } j \right\} 
\end{align}
We can reformulate this by counting how often the $2T+1$ different elements of $\Omega_0^{(i)}$ are present in a particular instance of the above sum, and get
\begin{align}
    \Omega^{(i)} &= \left\{ \left.\sum_{\oomega \in \Omega_0^{(i)}} N_{\oomega} \oomega \, \right| \, N_{\oomega} \geq 0, \sum_{\oomega \in \Omega_0^{(i)}} N_{\oomega} = N^{(i)} \right\}.
\end{align}
To bound the size of this set, we exploit the symmetry of the underlying frequencies $\delta_j \in \Delta(H^{(i)})$. Let us outline the idea: We will first count how we can distribute the number $N^{(i)}$ of repetitions over the different non-negative frequencies $\delta_j$ and then multiply this with the number of different frequencies that can be created by repeating $\delta_j$ and $-\delta_j$. 
To improve the scaling we get at the end, we will resort to a small trick and actually group the frequency $0$, which we know to always be present in the spectrum, with the first other frequency, therefore considering the combinatoric problem of distributing $N^{(i)}$ \enquote{balls} over $T$ distinguishable \enquote{bins} where some bins can be empty. The different possible ways to achieve this task are given by counting the \emph{weak compositions} of $N^{(i)}$ into $T$ parts, $\calC(N^{(i)},T)$. The number of such weak compositions is
\begin{align}
    |\calC(N^{(i)},T)| &= \begin{pmatrix}
    N^{(i)} + T -1\\ N^{(i)}
    \end{pmatrix} \\
    &= \frac{(N^{(i)}+T-1)!}{N^{(i)}! (T-1)!}\\
    & = \frac{(N^{(i)}+T-1)(N^{(i)}+T-2) \dots (N^{(i)}+1)}{(T-1)!}\\
    & =
    \calO((N^{(i)})^{T-1}).
\end{align}
We will denote such a composition as $(N_j^{(i)})_{j=1}^T \in \calC(N^{(i)}, T)$.
A simple counting argument reveals that there are $2N_1^{(i)} +1$ possible sums with $N_1^{(i)}$ elements from the set $\{0, \delta_1, -\delta_1\}$ and $N_j^{(i)} + 1 \leq 2N_j^{(i)} + 1$ possible sums with $N_j^{(i)}$ elements from the set $\{ \delta_j, -\delta_j\}$.
We can therefore bound
\begin{align}
    |\Omega^{(i)}|&\leq \sum_{(N_k^{(i)}) \in \calC(N^{(i)},T)} \prod_{k = 1}^{T} \left(2N_k^{(i)} + 1\right) \\
    &\leq \sum_{(N_k^{(i)}) \in \calC(N^{(i)},T)} \left(\frac{2\sum_{k=1}^T N_k^{(i)}}{T}+1\right)^T \\
    &\leq \sum_{(N_k^{(i)}) \in \calC(N^{(i)},T)} \left(\frac{2N^{(i)}}{T}+1\right)^T \\
    &=|\calC(N^{(i)},T)| \left(\frac{2N^{(i)}}{T}+1\right)^T \\
    &= \calO((N^{(i)})^{2T-1}),
\end{align}
where we have used the arithmetic-geometric mean inequality to obtain the second inequality. From this inequality, we see that by repeating the same Hamiltonian for an encoding, we obtain a polynomial scaling in the number of repetitions whose exponent depends on the number of different frequencies generated by the repeated Hamiltonian.

\textbf{Pauli encodings.}
Encodings performed with Hamiltonians that are a tensor product of Pauli operators, $H = \bigotimes_{k=1}^n P^{(k)}$ where $P^{(k)} \in \{ \bbI, X, Y, Z\}$, have been analyzed in Ref.~\cite{SchuldSwekeMeyer2021}. Therein, it was shown that $N^{(i)}$ repetitions of such encodings of arbitrary dimension will result in $|\Omega^{(i)}| = 2 N^{(i)}  +1$ and $\Omega^{(i)}\subseteq \ee^{(i)}\mathbb{Z}\subseteq\mathbb{Z}^d$.

\textbf{Summary.}
We can easily connect  the different upper bounds on $|\Omega^{(i)}|$ to upper bounds on $|\Omega|$ via the arithmetic-geometric mean inequality, i.e.,
\begin{align}
    \prod_{i=1}^d |\Omega^{(i)}| \leq \left(\sum_{i=1}^d \frac{ |\Omega^{(i)}|}{d}\right)^d,
\end{align}
and by noting that, for $q\geq 1$,
\begin{align}
    \sum_{i=1}^d \frac{(N^{(i)})^q}{d} \leq \frac{\left(\sum_{i=1}^d N^{(i)}\right)^q}{d} = \frac{N^q}{d}.
\end{align}
Table~\ref{tab:all_upper_bounds} summarizes the different upper bounds on $|\Omega^{(i)}|$ for individual parameters $x^{(i)}$ derived in this section as well as the associated bounds on $|\Omega|$.

\begin{table}[]
    \centering
    \caption{\label{tab:all_upper_bounds}Scaling of the different upper bounds for the number of different frequencies for the encoding of a single parameter $|\Omega^{(i)}|$, as well as the associated bounds for the scaling of the number of different frequencies for the total data-encoding strategy, $|\Omega|$. $N^{(i)}$ denotes the number of gates used for encoding the input $x^{(i)}$, $N$ denotes the total number of gates for all inputs.}
    \begin{tabularx}{\textwidth}{Xll}
        \toprule
        Encoding strategy & Upper bound on $|\Omega^{(i)}|$ & Upper bound on $|\Omega|$\\
        \midrule 
        Repetition of arbitrary Pauli encodings & $\displaystyle\calO\left(N^{(i)}\right)$& $\displaystyle\calO\left(\left(\frac{N}{d}\right)^d\right)$ \\ \addlinespace[0.25em]
        Repetition of the same encoding \newline with $2T+1$ frequencies & $\displaystyle\calO\left((N^{(i)})^{2T-1}\right)$ & $\displaystyle\calO\left(\left(\frac{N^{2T-1}}{d}\right)^d\right)$\\\addlinespace[0.25em]
        Repetition of the same $\kappa$-local encoding & $\displaystyle\calO\left((N^{(i)})^{2^{\kappa+1}-1}\right)$ & $\displaystyle\calO\left(\left(\frac{N^{2^{\kappa+1}-1}}{d}\right)^d\right)$\\\addlinespace[0.25em]
        Different $\kappa$-local encodings & $\displaystyle\calO\left(2^{2 \kappa N^{(i)}}\right)$ & $\displaystyle\calO\left(2^{2 \kappa N}\right)$\\
        \bottomrule
    \end{tabularx}
\end{table}

Given these results, we are finally in a position to provide a concrete answer to the first question posed in Section~\ref{s:motivation}. More specifically, by substituting the upper bounds on $|\Omega|$ given in Table \ref{tab:all_upper_bounds} into the generalization bounds for GTPs given in Section~\ref{s:results}, we can obtain generalization bounds for PQC-based model classes which depend explicitly on architectural hyper-parameters associated with the data-encoding strategy. Recall that we denoted the function class associated with a particular set of parameters $\Theta$, an encoding strategy $\calD$ and an observable $M$, as $\calF_{\Theta, \calD, M}$. 
We then obtain from Theorems \ref{theorem:rademacher-gen-bound-trig-polys} and \ref{theorem:covering-gen-bound-trig-polys} the following Corollary:

\begin{corollary}[Generalization bound for PQCs---From Theorems \ref{theorem:rademacher-gen-bound-trig-polys} and \ref{theorem:covering-gen-bound-trig-polys}]\label{corollary:gen-bound-PQCs}
Let $d,m\in\mathbb{N}$. Let $\ell:\R\times\R\to [0,c]$ be a bounded loss function such that~$\mathbb{R}\ni z\mapsto\ell(y,z)$ is $L$-Lipschitz for all $y\in\mathbb{R}$. For any $\delta\in (0,1)$ and for any probability measure $P$ on $[0,2\pi)^d\times \R$, with probability $\geq 1-\delta$ over the choice of i.i.d.~training data $S=\{(\xx_i,y_i)\}_{i=1}^m\in ([0,2\pi)^d\times\R)^m$ of size $m$ and every $f\in\calF_{\Theta,\calD,M}$, where $\calD$ is an encoding strategy with $N$ gates in total, we have that,
\begin{enumerate}[label=(\alph*)]
    \item if $\calD$ denotes any data-encoding strategy consisting of Hamiltonians that are tensor products of Pauli operators,
    \begin{align}
    R(f) - \hat{R}_S(f)
    \leq \tilde{\mathcal{O}}\left( \frac{L \lVert M \rVert_{\infty} }{\sqrt{m}}\left(\frac{N}{d}\right)^{\frac{d}{2}} + c\sqrt{\frac{\log \nicefrac{1}{\delta}}{m}}\right),
    \end{align}
    
    \item if $\calD$ denotes any data-encoding strategy consisting of the same single Hamiltonian with $T$ integer frequencies per data coordinate,
    \begin{align}
    R(f) - \hat{R}_S(f)
    \leq \tilde{\mathcal{O}}\left( \frac{L \lVert M \rVert_{\infty} }{\sqrt{m}}\left(\frac{N^{2T-1}}{d}\right)^{\frac{d}{2}} + c\sqrt{\frac{\log \nicefrac{1}{\delta}}{m}}\right),
    \end{align}
    
    \item if $\calD$ denotes any data-encoding strategy consisting of the same single $\kappa$-local Hamiltonian with integer frequencies per data coordinate,
    \begin{align}
        R(f) - \hat{R}_S(f)
    \leq \tilde{\mathcal{O}}\left( \frac{L \lVert M \rVert_{\infty} }{\sqrt{m}}\left(\frac{N^{2^{\kappa+1}-1}}{d}\right)^{\frac{d}{2}} + c\sqrt{\frac{\log \nicefrac{1}{\delta}}{m}}\right),
    \end{align}
    
    \item if $\calD$ denotes any data-encoding strategy consisting of possibly different $\kappa$-local Hamiltonians with integer frequencies per data coordinate,
    \begin{align}
        R(f) - \hat{R}_S(f)
    \leq \tilde{\mathcal{O}}\left(\frac{L \lVert M \rVert_{\infty} }{\sqrt{m}}\, 2^{\kappa N} + c\sqrt{\frac{\log \nicefrac{1}{\delta}}{m}}\right).
    \end{align}
\end{enumerate}
\end{corollary}
While we consider only four specific data-encoding strategies in this corollary, the generalization bounds from Theorems \ref{theorem:rademacher-gen-bound-trig-polys} and \ref{theorem:covering-gen-bound-trig-polys} can in principle be applied to PQC-based models with \emph{any} data-encoding strategy. To use the bounds, the corresponding $\lvert\Omega(\mathcal{D})\rvert$ has to be identified, which can then be readily combined with our generalization bounds for GTPs. 
Additionally, if Conjecture~\ref{conjecture:strengthened-inclusion} is correct, then the statements of Corollary~\ref{corollary:gen-bound-PQCs}(b)-(d) are true also for data-encoding Hamiltonians leading to non-integer frequencies.

\subsection{Comparison of data-encoding strategies from a generalization perspective}\label{sbs:comparison}

The results of the previous subsection give a concrete answer to the first question posed in Section~\ref{s:motivation}, namely explicitly encoding-dependent generalization bounds for PQC-based models.
However, recall from Section~\ref{s:motivation} that we also aimed to use such bounds to identify data-encoding strategies which give rise to a slow (polynomial) growth of model complexity with respect to increasingly complex data-encoding strategies, and therefore facilitate meaningful model selection via structural risk minimization. The results of the previous section now allow us to address this additional goal.

Given an assumption or constraint on the structure of the data-encoding Hamiltonians in a possible data-encoding strategy, 
the most natural data-encoding hyper-parameter for structural risk minimization is the number $N$ of encoding Hamiltonians. We see that using either repeated Pauli Hamiltonians, a repeated (but fixed) $\kappa$-local Hamiltonian, or the repetition of a fixed Hamiltonian with $2T+1$ frequencies, leads to a complexity bound and generalization bound that scale polynomially with $N$. 
However, using $N$ \emph{different} $\kappa$-local data-encoding Hamiltonians can lead, in the worst case, to complexity upper bounds which scale exponentially with respect to $N$. In the latter case we stress, however, that these worst-case bounds are constructed using Hamiltonians designed to saturate the maximum possible number of frequency differences, and in many cases the complexity scaling with respect to $N$ may be much slower. Additionally, while the polynomial generalization bounds we obtain for the first three data-encoding strategies give us hope in the possibility of meaningful structural risk minimization with respect to the number of data-encoding gates, our upper bounds on the generalization gap are not necessarily tight. Hence, we cannot rule out the possibility of better bounds for strategies consisting of many different Hamiltonians, which would facilitate the use of strucural risk minimization.

Additionally, while increasing the complexity of a data-encoding strategy by increasing $N$ is a natural (and experimentally feasible) strategy, in principle one might also consider increasing either the locality $\kappa$ or the number of frequencies $T$ of the repeated data-encoding Hamiltonian. This would be particularly relevant in the realistic scenario where experimental constraints severely limit the number of data-encoding gates which can be used. However, apart from the potential experimental obstacles one would face in doing so, we note that while our complexity bounds are polynomial with respect to $N$ (when keeping $\kappa$ and $T$ fixed), they are exponential (or doubly-exponential) with respect to $\kappa$ and $T$ respectively (when keeping $N$ fixed). As such, given the generalization bounds we have obtained in this work, from the generalization and structural risk minimization perspective it makes the most sense to systematically increase the complexity of the data-encoding strategy by keeping $\kappa$ and/or $T$ constant, and increasing the number of data-encoding gates. 

\section{Discussion}\label{s:discussion}

As discussed in Section~\ref{s:motivation}, the results from the previous section can be applied in a variety of ways. In particular, apart from the straightforward application of (probabilistically) bounding the generalization gap of an output hypothesis, or bounding the number of data samples required to guarantee an output hypothesis with a sufficiently small generalization gap, our results also facilitate the use of \textit{structural risk minimization} with respect to architectural hyper-parameters related to the data-encoding strategy. We reiterate that the results obtained here should be viewed as \textit{complementary} to many of the prior results discussed in Section~\ref{s:prior_work}. In particular, our results complement those which derive generalization bounds applicable to the same PQC-based hypothesis classes, but with explicit dependencies on architectural hyper-parameters which do not appear in our generalization bounds, such as depth, width, and total number of trainable gates.

More specifically, the generalization bounds of Section~\ref{s:gen-bounds-pqc} allow one to use structural risk minimization to find the optimal setting for data-encoding hyper-parameters (in the sense of yielding an output hypothesis with the smallest upper bound on true risk). However, they \textit{do not} give any guidance as to how one should choose the remaining architectural hyper-parameters, and in particular those related to the trainable parts of the PQC. 
As such, a natural (and recommended) strategy is to use different available and applicable generalization bounds to perform ``multi-dimensional structural risk minimization:" One can vary \textit{all} architectural hyper-parameters for which one has a generalization bound, and evaluate each hyper-parameter setting with respect to an upper bound on the true risk obtained from a union bound over all existing applicable bounds.
To make this more concrete, assume that we have a family of hypothesis classes $\{\mathcal{F}_{(k_1,k_2)}\}$, parametrized by two architectural hyper-parameters $k_1$ and $k_2$ (for example $k_1$ could be the number of encoding gates, and $k_2$ could be the number of trainable gates in a PQC
based model). Additionally, let us assume that we have derived two different generalization bounds, one depending on $k_1$, the other depending on $k_2$. More concretely, assume that we have a function $g_1(k_1,m,\delta)$ and a function $g_2(k_2,m,\delta)$ such that, for all $i\in\{1,2\}$, for all $\delta\in(0,1)$, with probability $1-\delta$ over $S\sim P^m$, for all $h\in\mathcal{F}_{(k_1,k_2)}$ we have that
\begin{align}
    R(h) &\leq \hat{R}_S(h) + g_i(k_i,m,\delta).
\end{align}
Using a union bound, we can then straightforwardly combine these two results to obtain the following generalization bound: For all $\delta\in(0,1)$, with probability $1-\delta$ over $S\sim P^m$, for all $h\in\mathcal{F}_{(k_1,k_2)}$ we have that
\begin{align}
    R(h) &\leq \hat{R}_S(h) + \min_{i}\left[g_i(k_i,m,\delta/2) \right].
\end{align}
We see that we can perform structural risk minimization by varying both $k_1$ and $k_2$ and using $\min_{i}\left[g_i(k_i,m,\delta/2)\right]$ to calculate an upper bound on the true risk of the candidate hypothesis. The above argument can clearly be generalized to an arbitrary number of architectural hyper-parameters, and thereby yields a methodology for exploiting multiple existing generalization bounds for ``multi-dimensional structural risk minimization.''

While the approach we have just discussed certainly allows us to exploit existing complementary generalization bounds depending on different architectural hyperparameters, it is an interesting open question whether one can derive generalization bounds which depend \textit{simultaneously} on multiple architectural hyper-parameters. In particular, it is of interest to understand whether one can in this way obtain generalization bounds, depending on multiple architectural hyper-parameters, which are \textit{tighter} than the bounds obtained by taking a union bound over existing bounds, each of which depends only on a single hyper-parameter. A potential strategy for obtaining such bounds would be to better understand the effect of structural assumptions on the trainable part of a PQC architecture on the structure of the \textit{coefficients} of the associated GTP representation. More concretely, while in this work we have focused on the frequency spectra of the GTPs, which are fully determined by the data-encoding strategy, the coefficients of the GTPs are determined by both the data-encoding strategy and the trainable part of the circuit. If one can characterize the implications of different circuit architectures on the structure of GTP coefficients, one could plausibly use refinements of the techniques presented in Section~\ref{s:results} to derive generalization bounds for the relevant GTPs that depend simultaneously on both the data-encoding strategy and complementary parameters of the circuit architecture. For example, certain PQC architectures may lead to GTP coefficients with a specific sparsity structure, or a constrained upper bound on a specific norm. 
Such a norm-specific bound
may allow us to exploit the general $p$-norm extensions of our GTP bounds, mentioned in Remarks~\ref{rmk:rademacher-p-norms} and \ref{rmk:covering-p-norms}, to derive generalization bounds which also depend on the trainable circuit architecture.

Finally, we recall the potential shortcomings of 
\textit{uniform} generalization bounds. In particular, in Ref.~\cite{zhang}, the authors have shown both experimentally and analytically that sufficiently complex neural networks can achieve zero empirical risk for classification tasks with randomly assigned labels. As the true risk for such a learning problem can be no better than what would be achieved by random guessing, any \textit{uniform} generalization bound for such a hypothesis class cannot offer any meaningful information in this complexity regime. More specifically, as uniform generalization bounds hold, by definition, for all hypotheses in the hypothesis class, and as there exist hypotheses which can achieve zero empirical risk even when generalization is not possible (i.e., when labels are selected randomly), such uniform bounds must be trivial.

It is, however, critical to emphasize that this finding applies only to \textit{sufficiently complex} hypothesis classes. More specifically, they apply to models capable of achieving zero empirical risk even for completely unstructured data, which typically requires that the number of model parameters is at least as large as the number of elements in the training data set. As the number of parameters in a NISQ-regime PQC-based model is typically orders of magnitude less than the size of training data sets associated with ``real-world'' learning problems, it is unlikely that these known issues with uniform generalization bounds hinder the application of our uniform bounds to 
the analysis of currently available and near-term PQC-based hypothesis classes.

Despite this, it is important to keep these concerns in mind as the complexity of available PQC-based models increases. 
Consequently, there are a variety of natural open questions for future research: Firstly, can one replicate both the experimental and analytical aspects of Ref.~\cite{zhang} for PQC-based model classes? This would help to determine whether (or when) it is necessary to move beyond uniform generalization bounds for PQC-based models. In particular, from an experimental perspective, can one demonstrate the ability of a (sufficiently complex) PQC-based model class to achieve zero risk for a randomly-relabeled real-world classification task? Secondly, can one put an analytical bound on what is 
\enquote{sufficiently complex}, i.e., how many model parameters are sufficient to ensure that for \textit{any} training data set of size $m$, there \textit{always} exists a hypothesis in the hypothesis class which can achieve zero empirical risk? Additionally, the shortcomings of uniform generalization bounds exposed in Ref.~\cite{zhang} have stimulated an explosion of research on non-uniform generalization bounds for highly complex neural network models~\cite{jiang}. It would be of interest to understand whether or how one can obtain non-uniform generalization bounds for PQC-based models, which would tighten the bounds obtained in this work in the future regime of high complexity.

\section{Conclusion}\label{s:conclusion}

In this work, we have derived Rademacher complexity and metric entropy bounds for PQC-based model classes. These depend explicitly on architectural hyper-parameters associated with the data-encoding strategy and are applicable to PQC-based models incorporating data re-uploading. By exploiting tools and techniques from statistical learning theory, we have then used these complexity bounds to obtain uniform generalization bounds, which allow to place a probabilistic upper-bound on the out-of-sample performance of any hypothesis, given its performance on the data. Additionally, we have used the obtained generalization bounds to compare data-encoding strategies from a generalization perspective and have discussed how, for certain data-encoding strategies, our generalization bounds may be used for model selection via structural risk minimization. We have stressed how the encoding-dependent generalization bounds obtained in this work should be viewed as \textit{complementary} to existing complexity and generalization bounds for PQC-based models, which depend explicitly on architectural hyper-parameters to which our bounds are insensitive. More specifically, we have sketched in Section~\ref{s:discussion} how the combination of our bounds with existing works facilitates model selection via multi-dimensional structural risk minimization. Finally, as discussed in Section~\ref{s:discussion}, it is important to acknowledge that the bounds we have obtained here are expected to be useful for PQC-based models in the ``moderate-complexity" regime, i.e., for models parametrized by fewer parameters than the number of available data samples. However, in analogy with known results for classical model classes, these bounds may cease to be meaningful as the complexity of PQC-based models increases into an over-parametrized regime. Given this, we have also sketched in Section~\ref{s:discussion} a variety of open questions and directions for future research.

\begin{acknowledgments}
The authors would like to thank Alexander Nietner for insightful discussions and Maria Schuld and David Sutter for helpful feedback on an earlier draft.
Moreover, the authors thank Stefano Mangini for pointing out a mistake in an earlier version of this paper, in which we erroneously claimed that Eq.~\eqref{e:gtp_coeff_bound} holds true for arbitrary frequency spectra.
We would like to thank the Cluster of Excellence MATH+ (EF1-11), the BMWi (PlanQK), for which this work provides an understanding of models of quantum-enhanced machine learning,
and the BMBF (Hybrid), for which this work 
helps with the design of
quantum-classical hybrid models of quantum 
computing, for support. This work has 
also been supported by the DFG (CRC 183, 
project B01), the Einstein Foundation 
(Einstein Research Unit on quantum devices) 
and by the EU's Horizon 2020 research and innovation programme under grant agreement No.~817482 (PASQuanS). M.C.C.~gratefully acknowledges support from the TopMath Graduate Center of the TUM Graduate School at the Technical University of Munich, Germany, from the TopMath Program at the Elite Network of Bavaria, and from the German Academic Scholarship Foundation (Studienstiftung des deutschen Volkes).
\end{acknowledgments}


\appendix

\section{Auxiliary results from statistical learning theory}\label{appendix:auxiliary-results}

In this appendix, we collect some well known results from classical statistical learning theory that we make use of in our proofs.

\begin{lemma}[Rademacher complexity progression (Theorem 2.15 in Ref.~\cite{Wolf.2020})]\label{thm:2.15Wolf}
    Let $a, b\in\R$ and $\Tilde{\sigma}:\R\to\R$ an $L$-Lipschitz function and assume $\calF_0\subseteq\R^{\calX}$ is a set of functions that includes the $0$ function.
    Also, let $\calF$ be the following function class
    \begin{align}
        \calF & \coloneqq \Bigg\{\xx\mapsto\Tilde{\sigma}\left(v + \sum_{j=1}^{m}\omega_jf_j(\xx)\right) \,\Bigg|\, \lvert v\rvert \leq a, \norm{\oomega}_1\leq b, \text{ and } f_j\in\calF_0\Bigg\}.
    \end{align}
    Then, the empirical Rademacher complexity of $\calF$ with respect to any point $\Vec{\xx}\in\calX^m$ can be bounded in terms of the one of $\calF_0$
    \begin{align}
        \empRad_{\Vec{\xx}}(\calF) \leq L\left(\frac{a}{\sqrt{m}} + 2b\empRad(\calF_0)\right).
    \end{align}
    The $2$ factor can be dropped if $\calF_0 = -\calF_0$.
\end{lemma}

\begin{lemma}[Rademacher complexity of layered network (Corollary 2.11 in Ref.~\cite{Wolf.2020})]\label{cor:2.11Wolf}
    Let $a,b>0$ and $\calX\coloneqq\left\{\xx\in\R^d\,|\,\norm{\xx}_\infty\leq C\right\}$.
    Consider a neural network architecture with $\delta$ hidden layers that implements $\calF\subseteq\R^{\calX}$, and such that
    \begin{enumerate}
        \item The activation function $\sigma:\R\to\R$ is $L$-Lipschitz and anti-symmetric.
        \item For every neuron, the vector of weights $\oomega$ satisfies $\norm{\oomega}_1\leq b$.
        \item For every neuron, the modulus of the bias is upper-bounded by $a$.
    \end{enumerate}
    Then, the empirical Rademacher complexity of $\calF$ with respect to any point $\Vec{\xx}\in\calX^m$ can be upper-bounded as
    \begin{align}
        \empRad_{\Vec{\xx}}(\calF) \leq\frac{1}{\sqrt{m}}\left(Cb^\delta\sqrt{2\log(2d)} + a\sum_{i=0}^{\delta-1}b^i\right).
    \end{align}
\end{lemma}

\begin{lemma}[Massart's Lemma~\cite{Massart.2000}]\label{lemma:Massart}
Let $N\in\mathbb{N}$. Let $A\subset\mathbb{R}^N$ be a finite set contained in a Euclidean ball of radius $r>0$. Then
\begin{align}
    \mathbb{E}_\sigma \left[\sup\limits_{a\in A}\frac{1}{n}\sum\limits_{i=1}^N \sigma_i a_i\right]
    &\leq \frac{r\sqrt{2\log\lvert A\rvert}}{N},
\end{align}
where the expectation is with respect to i.i.d.~Rademacher random variables $\sigma_1,\ldots,\sigma_N$.
\end{lemma}

\begin{lemma}[Talagrand's Lemma (going back to~\cite{ledoux1991probability}; see also Lemma $5.7$ in Ref.~\cite{MohriRostamizadehTalwalkar18})]\label{lemma:talagrand}
Let $\ell_1,\ldots,\ell_m:\mathbb{R}\to\mathbb{R}$ be $L$-Lipschitz functions. Let $\mathcal{F}\subset\mathbb{R}^\mathcal{Z}$ be a class of real-valued functions on some data space $\mathcal{Z}$. Then, for any $\zz_1,\ldots,\zz_m\in\mathcal{Z}$, 
\begin{align}
    \frac{1}{m}\mathbb{E}_\sigma\left[\sup\limits_{f\in\mathcal{F}} \sum\limits_{i=1}^m \sigma_i \ell\circ f(\zz_i)\right]
    &\leq \frac{L}{m}\mathbb{E}_\sigma\left[\sup\limits_{f\in\Fclass} \sum\limits_{i=1}^m \sigma_i f(\zz_i)\right],
\end{align}
where the expectations are over i.i.d.~Rademacher random variables $\sigma_1,\ldots,\sigma_m$.
\end{lemma}

\begin{theorem}[Dudley's Theorem (\cite{dudley_1999}; see also Theorem $8.1.2$ in Ref.~\cite{Vershynin.2018} or Theorem $1.19$ in Ref.~\cite{Wolf.2020})]\label{theorem:dudley}
For a fixed vector $z\in\mathcal{Z}^m$ let $\mathcal{G}$ be a subset of the pseudo-metric space $(\mathbb{R}^\mathcal{Z}, \norm{\cdot}_{2,z})$ and let $\gamma_0\coloneqq\sup_{g\in\mathcal{G}}\norm{g}_{2,z}$. Then the empirical Rademacher complexity $\hat{\mathcal{R}}_z(\mathcal{G})$ of $\mathcal{G}$ with respect to~$z$ can be upper-bounded as
\begin{align}
    \hat{\mathcal{R}}_z(\mathcal{G})
    &\leq \inf\limits_{\varepsilon\in[0,\frac{\gamma_0}{2})}\left\{4\varepsilon + \frac{12}{\sqrt{m}}\int\limits_{\varepsilon}^{\gamma_0}\sqrt{\log \mathcal{N}(\mathcal{G}, \norm{\cdot}_{2,z}, \beta)}~\mathrm{d}\beta\right\}.
\end{align}
\end{theorem}

\newpage
\printbibliography
\end{document}